\documentclass[prd,tightenlines,nofootinbib,superscriptaddress]{revtex4}

\usepackage{amsfonts,amssymb,amsthm,bbm}

\usepackage{amsmath}

\allowdisplaybreaks[3]

\usepackage{hyperref}

\usepackage{subcaption}
\usepackage{color,psfrag}
\usepackage[dvips]{graphicx}

\usepackage[dvipsnames]{xcolor}

\usepackage{tikz}
\usetikzlibrary{calc,matrix}
\usetikzlibrary{arrows,decorations.pathmorphing,decorations.pathreplacing,decorations.markings}

\tikzset{
    edge/.style={draw, postaction={decorate},
        decoration={markings,mark=at position .55 with {\arrow{>}}}},
    linking-D/.style={draw, postaction={decorate},
        decoration={markings,mark=at position 1 with { rectangle, draw, inner sep=1pt, minimum size=2mm, fill=cyan }}},
}

\usepackage{mathtools}

\newcounter{footnotemarknum}

\newcommand{\C}{{\mathbb C}}
\newcommand{\N}{{\mathbb N}}
\newcommand{\R}{{\mathbb R}}

\newcommand{\cB}{{\mathcal B}}
\newcommand{\cE}{{\mathcal E}}

\newcommand{\cL}{{\mathcal L}}
\newcommand{\cH}{{\mathcal H}}
\newcommand{\cM}{{\mathcal M}}
\newcommand{\cN}{{\mathcal N}}
\newcommand{\cR}{{\mathcal R}}

\newcommand{\cP}{{\mathcal P}}

\newcommand{\cV}{{\mathcal V}}
\newcommand{\cD}{{\mathcal D}}

\newcommand{\cS}{{\mathcal S}}

\newcommand{\cI}{{\mathcal I}}

\newcommand{\SU}{\mathrm{SU}}

\newcommand{\be}{\begin{equation}}
\newcommand{\ee}{\end{equation}}
\newcommand{\beq}{\begin{eqnarray}}
\newcommand{\eeq}{\end{eqnarray}}
\newcommand{\bes}{\begin{eqnarray}}
\newcommand{\ees}{\end{eqnarray}}

\renewcommand{\u}{{\mathfrak{u}}}
\newcommand{\su}{{\mathfrak{su}}}

\newcommand{\la}{\langle}
\newcommand{\ra}{\rangle}

\newcommand{\tr}{{\mathrm{Tr}}}
\newcommand{\f}{\frac}

\def\nn{\nonumber}
\def\pp{\partial}

\def\rd{\mathrm{d}}

\def\vphi{\varphi}
\def\eps{\epsilon}

\newcommand{\id}{\mathbb{I}}

\def\tk{\tilde{k}}
\def\tm{\tilde{m}}

\def\bGamma{{\Gamma}^{\circ}}

\def\bI{{}^{\pp}I}
\def\BI{{}^{o}I}
\def\tiota{\tilde{\iota}}

\newtheorem{theorem}{Theorem}[section]
\newtheorem{lemma}[theorem]{Lemma}
\newtheorem{prop}[theorem]{Proposition}

\newcommand*\circled[1]{\tikz[baseline=(char.base)]{
            \node[shape=circle,draw,inner sep=0.7pt] (char) {#1};}}

\begin{document}

\title{Loop Quantum Gravity's Boundary Maps}

\author{{\bf Qian Chen}}\email{chenqian.phys@gmail.com}
\affiliation{Universit\'e de Lyon, ENS de Lyon, CNRS, Laboratoire de Physique LPENSL, 69007 Lyon, France}
\author{{\bf Etera R. Livine}}\email{etera.livine@ens-lyon.fr}
\affiliation{Universit\'e de Lyon, ENS de Lyon, CNRS, Laboratoire de Physique LPENSL, 69007 Lyon, France}

\date{\today}

\begin{abstract}

In canonical quantum gravity, the presence of spatial boundaries naturally leads to a boundary quantum states, representing quantum boundary conditions for the bulk fields. As a consequence, quantum states of the bulk geometry needs to be upgraded to wave-functions valued in the boundary Hilbert space: the bulk become quantum operator acting on boundary states. We apply this to loop quantum gravity and describe spin networks with 2d boundary as wave-functions mapping bulk holonomies to spin states on the boundary. This sets the bulk-boundary relation in a clear mathematical framework, which allows to define the boundary density matrix induced by a bulk spin network states after tracing out the bulk degrees of freedom.
We ask the question of the bulk reconstruction and prove a boundary-to-bulk universal reconstruction procedure, to be understood as a purification of the mixed boundary state into a pure bulk state.
We further perform a first investigation in the algebraic structure of induced boundary density matrices and show how correlations between bulk excitations, i.e. quanta of 3d geometry, get reflected into the boundary density matrix.
\end{abstract}

\maketitle

\tableofcontents

\section*{Introduction}

%
%


Loop quantum gravity is one of the main frameworks attempting to quantize general relativity in a non-perturbative way and, in doing so, define a background independent theory of quantum gravity (for reviews, see \cite{Gaul:1999ys,Thiemann:2007pyv,Rovelli:2014ssa,Bodendorfer:2016uat}). Based on a first order reformulation of general relativity \`a la Cartan, it trades the 4-metric for  vierbein-connection fields and writes general relativity as a gauge field theory defined by the Holst-Palatini action \cite{Holst:1995pc}. Focusing on a Hamiltonian formulation of the theory, it proceeds to a 3+1 decomposition of space-time and studies the evolution in time of the space geometry. The geometry of 3d space slices is described by a pair of canonical fields, the (co-)triad and the Ashtekar-Barbero connection, which enhance the 3-metric and extrinsic curvature of the Arnowitt-Deser-Misner (ADM) formalism with a local gauge invariance under $\SU(2)$ transformations (i.e. local 3d rotations in the tangent space).
The goal is then to provide a quantization of the (suitably deformed) Dirac algebra of Hamiltonian constraints generating space-time diffeomorphisms, by
representing (a suitable algebra of observables of) the triad-connection fields on a Hilbert space carrying an action of the (suitably deformed) space-time diffeomorphisms.

In this spirit, the standard loop quantum gravity approach performs a canonical quantization of the holonomy-flux algebra, of observables smearing the Ashtekar-Barbero connection along 1d curves and the (co-)triad along 2d surfaces, and defines quantum states of geometry as polymer structures or graph-like geometries. Those spin network states represent the excitations of the Ashtekar-Barbero connection as Wilson lines in gauge field theory.
Geometric observables are raised to quantum operators acting on the Hilbert space spanned by spin networks, leading to the celebrated result of a discrete spectra for  areas and volumes \cite{Rovelli:1994ge,Ashtekar:1996eg,Ashtekar:1997fb}.

Spin networks are actually the kinematical states of the theory and the game is to describe their dynamics, i.e. their evolution in time generated by the Hamiltonian (constraints). Although a traditional point of view is to attempt to discretize, regularize and quantize the Hamiltonian constraints \cite{Thiemann:1996aw,Thiemann:1996av}, this often leads to anomalies. The formalism naturally evolved towards a path integral formulation. The resulting spinfoam models,  constructed from (extended) topological quantum field theories (TQFTs) with defects, define transition amplitudes for histories of spin networks \cite{Reisenberger:1996pu,Baez:1997zt,Barrett:1997gw,Freidel:1998pt} (see \cite{Livine:2010zx,Dupuis:2011mz,Perez:2012wv} for reviews). The formalism then evolves in a third quantization, where so-called ``group field theories'' define non-perturbative sums over random spin network histories in a similar way than matrix model partition functions define sums over random 2d discrete surfaces \cite{DePietri:1999bx,Reisenberger:2000zc,Freidel:2005qe} (see \cite{Oriti:2006se,Carrozza:2013mna,Oriti:2014uga} for reviews).

Here we take a trip back to the foundations of loop quantum gravity, to describe (spatial) boundaries. Indeed, despite a whole branch of research dedicated to the study of (quantum) black holes and (isolated) horizon boundary conditions, most work in loop quantum gravity focuses on closed space (often implicitly done by studying spin network based on closed graphs). This focus translates the idea  of the universe as a closed system with subsystems interacting with each other and its translation into the definition of a wave-function for the entire universe as done in quantum cosmology.
However, the key function now played by the holographic principle as a guide for quantum gravity has put great emphasis of the role of boundaries. Although holography, inspired from black hole entropy and the AdS/CFT correspondence, can be initially thought as an asymptotic global property, recent researches on local area-entropy relations, holographic entanglement, holographic diamonds and the investigation of quasi-local holography and gravitational edge modes for finite boundaries necessarily pushes us to include (spatial) boundaries in the description of quantum geometries, not just as mere classical boundary conditions but as legitimate quantum boundary states. This translates a shift of perspective from a global description of space(-time) as a whole to a quasi-local description where any local bounded region of space(-time) is thought as an open quantum system.

To be more concrete, the geometrical setting we wish to study is a cylinder in space-time: we consider a bounded region of space $\cR$, with the topology of a 3-ball, whose boundary $\cS=\pp\cR$ has the topology of a 2-sphere; the space-time structure is then the cylinder $\cR\times [t_{i},t_{f}]$  whose time-like boundary is the 2+1-dimensional $\cB=\cS\times[t_{i},t_{f}]$, such that the space boundary can be considered as the corner of space-time $\cS=\cB\cap\cR_{i}$, as illustrated on fig.\ref{fig:corner}. A canonical framework describes the evolution in time of the state of the 3d geometry of the space slice $\cR$. In this context, the question of holography amounts to identify the degrees of freedom of the boundary geometry on the corner $\cS$ - the gravitational edge modes\footnotemark- which will generate the boundary conditions on $\cB$ for the bulk geometry,  study how the dynamics of those edge modes propagate into the bulk and, as a consequence, understand to which extent boundary observables reflect the bulk geometry's evolution and fluctuations.
\footnotetext{
For recent works on classical edge modes for general relativity in its first order formulation in terms of connection-vierbein variables, the interested reader can see \cite{Freidel:2019ees,Freidel:2020xyx,Freidel:2020svx,Freidel:2020ayo}.}
From this perspective, the study of holography is intimately intertwined with the renormalization flow \`a la Wilson, where the coarse-graining of the dynamics of the bulk geometry in $\cR$ induces  effective dynamics and boundary theory on $\cS$, in a bulk-to-boundary process which should ultimately be dual to the boundary-to-bulk reconstruction intended by holography (see e.g. \cite{Livine:2017xww} for an early attempt to realize this scenario in loop quantum gravity).
\begin{figure}[htb]
	\centering
	\begin{tikzpicture} []

\draw[->] (-2,-0.5) -- (-2,3.5) node[above] {$t$};

\coordinate  (O1) at (0,0);
\coordinate  (O2) at (2,0.7);
\coordinate  (O3) at (4,0);
\coordinate  (O4) at (2,-1);

\coordinate  (P1) at (0.8,3);
\coordinate  (P2) at (2.7,3.7);
\coordinate  (P3) at (5.3,3);
\coordinate  (P4) at (2.6,2);

\draw[dashed,color=blue] (O1) node[left,black] {$t_i$} to[out=85,in=160] (O2);
\draw[dashed,color=blue] (O2) to[out=-20,in=95] (O3);
\draw[color=blue]        (O3) to[out=-95,in=2] (O4);
\draw[color=blue]        (O4) to[out=178,in=-80] (O1);

\draw[color=blue] (P1) node[left,black] {$t_f$} to[out=90,in=160] (P2);
\draw[color=blue] (P2) to[out=-20,in=95] (P3);
\draw[color=blue]        (P3) to[out=-90,in=2] (P4);
\draw[color=blue]        (P4) to[out=178,in=-80] (P1);

\draw        (O1) to[out=95,in=-100] (P1);
\draw        (O3) to[out=80,in=-87] node[midway,below right] {$\cB=\cS\times[t_i,t_f]$} (P3);

\node at (2,0) {$\cR$};
\node at (4.5,-0.5) {$\cS=\pp\cR$};



	\end{tikzpicture}
	\caption{Boundary and  corner:
	we consider the evolution in time of a bounded region of space $\cR$ whose spatial boundary  $\cS=\pp\cR$ defines what is called the two-dimensional corner of space-time; the evolution of the corner defines the 2+1-d boundary of the region of space-time, $\cB=\cS\times [t_{i},t_{f}]$. 
	}
	\label{fig:corner}
\end{figure}
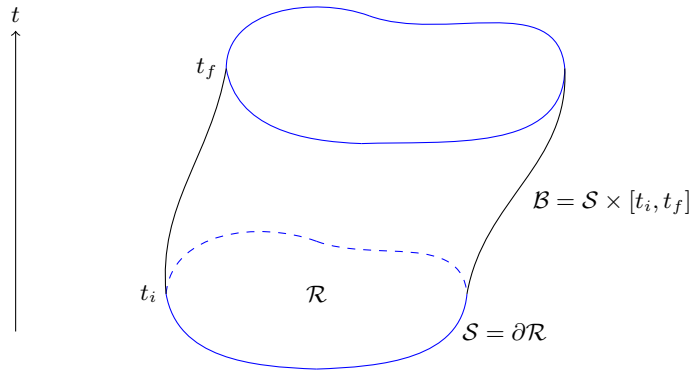

To implement this in quantum gravity, we follow a logic paralleling the hierarchy of  4d/3d/2d/1d defects and their algebraic description in a 4d TQFT, the introduction of quantum states on the boundary forces to go one level higher algebraically and define bulk states as operators (linear forms) acting on boundary states: bulk states will not simply be wave-functions valued in $\C$ but valued in  Hilbert space of boundary states.
To make things explicit, we call the boundary Hilbert space $\cH_{\pp}$ with boundary states $|\Phi_{\pp}\ra$ living on the space-time corner $\cS=\pp\cR$. A wave-function $\psi$ is a function of the bulk fields $\vphi_{bulk}$ valued in the (dual of the) boundary Hilbert space, $\psi[\vphi_{bulk}]\in\cH_{\pp}^{(*)}$, and thus defines a linear form on the boundary Hilbert space:
\be
\psi:\vphi_{bulk}\mapsto\Psi[\vphi_{bulk}] \in \cH_{\pp}^{(*)}
\,,\qquad
\la \psi[\vphi_{bulk}]\,|\,\Phi_{\pp}\ra\in \C\,.
\ee
One can then go two ways. Either we interpret these bulk wave-functions as defining a probability distribution for the bulk observables dependant on the choice of boundary states (i.e. quantum boundary conditions): once $\Phi_{\pp}$ is fixed, the function
\be
\la \Phi_{\pp}| \psi[\cdot]\ra:\,\vphi_{bulk}\mapsto \la \Phi_{\pp}| \psi[\vphi_{bulk}]\ra\in\C
\ee
is a standard $\C$-valued wave-function for the bulk fields.
Or we reverse this logic and look at the probability distribution for the boundary observables after integration over the bulk fields. In that case,
\be
\rho_{\pp}[\psi]=\int [\cD \vphi_{bulk}]\,|\psi[\vphi_{bulk}]\ra\la \psi[\vphi_{bulk}]|\in\,\textrm{End}[\cH_{\pp}]
\ee
is the density matrix induced on the boundary by the bulk state $\psi$.
The goal of this paper is to study the latter case in the framework of  loop quantum gravity and clearly  define this bulk-to-boundary coarse-graining from bulk spin networks to boundary density matrix.
This entails extending the spin network states of the 3d bulk geometry in $\cR$ to include the boundary degrees of freedom on the corner $\cS$. As we explain in the present paper, this can done in a natural way in loop quantum gravity since spin networks can be geometrically interpreted as aggregates of area quanta, glued together to create a 3d spaces from 2d excitations, and can thus be naturally extended to include the area quanta on the 2d boundary $\cS$.
A spin network wave-function on an open graph then naturally define a linear form on the Hilbert space of spin states living on the open edges of the graph, as illustrated on fig.\ref{fig:boundary} and thus induces a boundary density matrix.
\begin{figure}[htb]
	\centering
	\begin{tikzpicture} []

\coordinate  (O1) at (0,0);
\coordinate  (O2) at (2,0.7);
\coordinate  (O3) at (4,0);
\coordinate  (O4) at (2,-1);

\coordinate  (P1) at (0.8,3);
\coordinate  (P2) at (2.7,3.7);
\coordinate  (P3) at (5.3,3);
\coordinate  (P4) at (2.6,2);

\draw[dashed,color=blue] (O1) to[out=85,in=160] (O2);
\draw[dashed,color=blue] (O2) to[out=-20,in=95] (O3);
\draw[color=blue]        (O3) to[out=-95,in=2] (O4);
\draw[color=blue]        (O4) to[out=178,in=-80] (O1);

\draw[color=blue] (P1) to[out=90,in=160] (P2);
\draw[color=blue] (P2) to[out=-20,in=95] (P3);
\draw[color=blue]        (P3) to[out=-90,in=2] (P4);
\draw[color=blue]        (P4) to[out=178,in=-80] (P1);

\draw        (O1) to[out=95,in=-100] (P1);
\draw        (O3) to[out=80,in=-87] (P3);

\coordinate  (A1) at (0.4,0);
\coordinate  (A2) at (1.8,0.5);
\coordinate  (A3) at (3.4,0.2);
\coordinate  (A4) at (2.2,-0.6);
\coordinate  (A5) at (1.2,-0.6);

\draw (A1) [color=red] -- ++ (-0.6,-1);
\draw (A2) [color=red] -- ++ (-0.6,0.8);
\draw (A3) [color=red] -- ++ (0.6,0.8);
\draw (A4) [color=red] -- ++ (0.2,-0.95);
\draw (A5) [color=red] -- ++ (-0.3,-0.85);

\draw (A1) [color=green] -- (A2);
\draw (A2) [color=green] -- (A3);
\draw (A3) [color=green] -- (A4);
\draw (A4) [color=green] -- (A5);
\draw (A5) [color=green] -- (A1);
\draw (A2) [color=green] -- (A4);
\draw (A2) [color=green] -- (A5);

\node[scale=0.7,color=red] at (A1) {$\bullet$};
\node[scale=0.7,color=red] at (A2) {$\bullet$};
\node[scale=0.7,color=red] at (A3) {$\bullet$};
\node[scale=0.7,color=red] at (A4) {$\bullet$};
\node[scale=0.7,color=red] at (A5) {$\bullet$};

\coordinate  (B1) at (1.5,3.2);
\coordinate  (B2) at (2.6,3.5);
\coordinate  (B3) at (4.4,3.3);
\coordinate  (B4) at (3.2,2.4);
\coordinate  (B5) at (2.2,2.4);

\draw (B1) [color=red] -- ++ (-0.9,0.8);
\draw (B2) [color=red] -- ++ (-0.3,0.8);
\draw (B3) [color=red] -- ++ (0.6,0.8);
\draw (B4) [color=red] -- ++ (0.2,-0.9);
\draw (B5) [color=red] -- ++ (-0.3,-0.9);

\draw (B1) [color=green] -- (B2);
\draw (B2) [color=green] -- (B3);
\draw (B3) [color=green] -- (B4);
\draw (B4) [color=green] -- (B5);
\draw (B5) [color=green] -- (B1);
\draw (B2) [color=green] -- (B4);
\draw (B2) [color=green] -- (B5);

\node[scale=0.7,color=red] at (B1) {$\bullet$};
\node[scale=0.7,color=red] at (B2) {$\bullet$};
\node[scale=0.7,color=red] at (B3) {$\bullet$};
\node[scale=0.7,color=red] at (B4) {$\bullet$};
\node[scale=0.7,color=red] at (B5) {$\bullet$};

	\end{tikzpicture}
	\caption{Spin network with a boundary: on each spatial slice, the embedded graph $\Gamma$ punctures the boundary surface of the bounded region of space $\cR$; we distinguish the boundary edges $e\in \pp\Gamma$ in red and the bulk edges $e\in\bGamma$ in green; the spin network defines a wave-function for the holonomies living on the bulk edges valued in the Hilbert space attached to the open ends of the boundary edges.
	}
	\label{fig:boundary}
\end{figure}
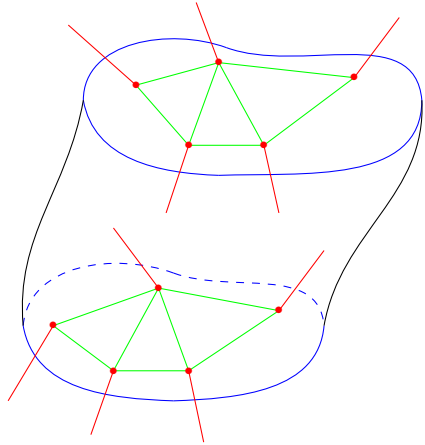
%


%



\medskip


The first section of this paper starts with a quick review of spin network quantum states for the 3d bulk geometry in loop quantum gravity. Then adapting this definition to the case of spatial boundaries, represented as open edges, we show that spin network wave-functions are actually valued in the boundary Hilbert space, i.e. they are functions of bulk $\SU(2)$ holonomies mapping them onto boundary spin states. The boundary spin states can be understood as quantum boundary conditions. In operational terms, the bulk spin network can be interpreted as a quantum circuit acting on the boundary spins. 
This opens the door to two directions. Either we sum over boundary states and obtain the probability distribution for the bulk holonomy. Or we can integrate over the bulk holonomies and obtain the {\it boundary density matrix} defining the distribution of boundary states induced by the bulk spin network. This boundary density matrix can be interpreted as a bulk-to-boundary coarse-graining of the spin network state of quantum geometry.

Section II is dedicated to the analysis of boundary density matrices induced by spin network states on fixed graphs and to a first study of their algebraic structure and properties. Our most important result is a universal bulk reconstruction procedure: starting from a gauge-invariant density matrix on the boundary Hilbert space, we show that one can always obtain it as the induced boundary density matrix of a spin network state on the bulk graph with a single vertex connected to all the boundary edges and to a single bulk loop. This can be understood as a purification result, since it shows how an arbitrary gauge-invariant mixed state on the boundary can be lifted to a pure bulk spin network state.
We then go on investigating the finer structure of the induced boundary density matrices in terms of boundary vertices and bouquets of boundary edges.

Section III finally presents explicit examples with the candy graphs, made of two vertices connected with bulk links, with four boundary edges and then with six boundary edges. This illustrates the various levels of mixed states one can obtain on the boundary in loop quantum gravity.

\section{Spin Networks as Boundary Maps}

For globally hyperbolic four-dimensional space-times $\cM=\Sigma\times\R$ with closed three-dimensional spatial slices $\Sigma$,
loop quantum gravity (LQG) defines quantum states of geometry and describes their constrained evolution in time. A state of 3d geometry are defined by  a closed oriented graph $\Gamma$ and a wave-function $\psi$ on it. This wave-function depends on one $\SU(2)$ group element for each edge $e$ of the graph, $g_{e}\in\SU(2)$, and is assumed to be invariant under the $\SU(2)$-action at each vertex $v$ of the graph:
\be \label{gauge transformation}
\psi(\{g_{e}\}_{e\in\Gamma})
=
\la\{g_{e}\}_{e\in\Gamma} | \psi\ra
=
\psi(\{h_{t(e)}g_{e}h_{s(e)}^{-1}\}_{e\in\Gamma})\,\quad
\forall h_{v}\in\SU(2)\,
\ee
where $t(e)$ and $s(e)$ respectively refer to the target and source vertices of the edge $e$. We write $E$ and $V$ respectively for the number of edges and vertices of the considered graph $\Gamma$.
The scalar product between such wave-functions is given by the Haar measure on the Lie group $\SU(2)$:
\be
\la \psi|\widetilde{\psi}\ra
=\int_{\SU(2)^{{\times E}}}\prod_{e}\rd g_{e}\,
\overline{\psi(\{g_{e}\}_{e\in\Gamma})}\,\widetilde{\psi}(\{g_{e}\}_{e\in\Gamma})
\,.
\ee
The Hilbert space of quantum states with support on the graph $\Gamma$ is thus realized as a space of square-integrable functions, $\cH_{\Gamma}=L^{2}(\SU(2)^{{\times E}}/\SU(2)^{{\times V}})$.

A basis of this Hilbert space can be constructed using the spin decomposition of $L^{2}$ functions on the Lie group $\SU(2)$ according to the Peter-Weyl theorem. A {\it spin} $j\in\f\N2$ defines an irreducible unitary representation of $\SU(2)$, with the action of $\SU(2)$ group elements realized on a $(2j+1)$-dimensional Hilbert space $\cV_{j}$. We use the standard orthonormal basis $|j,m\ra$, labeled by the spin $j$ and the magnetic index $m$ running by integer steps from $-j$ to $+j$, which diagonalizes the $\SU(2)$ Casimir $\vec{J}^{2}$ and the $\u(1)$ generator $J_{z}$. Group elements $g$ are then represented by the $(2j+1)\times (2j+1)$ Wigner matrices $D^{j}(g)$:
\be
D^{j}_{mm'}(g)=\la j,m|g|j,m'\ra\,,\qquad
\overline{D^{j}_{mm'}(g)}
=
D^{j}_{m'm}(g^{-1})
\,.
\ee
These Wigner matrices form an orthogonal basis of $L^{2}(\SU(2))$:
\be
\int_{\SU(2)}\rd g\,\overline{D^{j}_{ab}(g)}\,{D^{k}_{cd}(g)}
=
\int_{\SU(2)}\rd g\,\overline{D^{j}_{ba}(g^{-1})}\,{D^{k}_{cd}(g)}
=
\f{\delta_{jk}}{2j+1}\delta_{ac}\delta_{bd}
\,,\qquad
\delta(g)
=\sum_{j\in\f\N2}(2j+1)\chi_{j}(g)
\,, \label{eq:Peter-Weyl}
\ee
where $\chi_{j}$ is the spin-$j$ character defined as the trace of the Wigner matrix, $\chi_{j}(g)=\tr D^{j}(g)=\la j,m|g|j,m\ra$.
Applying this to gauge-invariant wave-functions allows to build the {\it spin network} basis states of $\cH_{\Gamma}$, which depend on one spin $j_{e}$ on each edge and one intertwiner $I_{v}$ at each vertex:
\be
\Psi_{\{j_{e},I_{v}\}}(\{g_{e}\}_{e\in\Gamma})
=
\la\{g_{e}\}) | \{j_{e},I_{v}\}\ra
=
\sum_{m_{e}^{t,s}}
\prod_{e}\sqrt{2j_{e}+1}\,\la j_{e}m_{e}^{t}|g_{e}|j_{e}m_{e}^{s}\ra
\,\prod_{v} \la \bigotimes_{e|\,v=s(e)} j_{e}m_{e}^{s}|\,I_{v}\,|\bigotimes_{e|\,v=t(e)}j_{e}m_{e}^{t}\ra
\,.
\ee
\begin{figure}[!htb]
	\centering
	\begin{tikzpicture} [scale=1.2]
\coordinate  (O) at (0,0);

\node[scale=0.7] at (O) {$\bullet$} node[below] {$I_v$};

\draw (O)  --  node[midway,sloped]{$>$} ++ (1,1) node[right] {$j_1, m_1$};

\draw (O)  to[bend left=20]  node[midway,sloped]{$<$} ++ (0.9,-0.9) node[right] {$j_2, m_2 $};

\draw (O)  to[bend left=20] node[midway,sloped]{$>$} ++ (0,1.5) node[above] {$j_5, m_5$};

\draw (O)  to[bend left=10]  node[midway,sloped]{$<$} ++ (-1.2,-0.5) node[left] {$j_3, m_3$};

\draw (O)  to[bend left=10]  node[midway,sloped]{$>$} ++ (-1.1,0.6) node[left] {$j_4, m_4$};

	\end{tikzpicture}
	\caption{A five-valent intertwiner $I_v$ at vertex $v$ is a $\SU(2)$-invariant map from the tensor product of the incoming spins (living on the edges $e$ whose target is $v$) to the outgoing spins (living on the edges $e$ whose source is $v$), its matrix elements are  $\la (j_{1},m_{1})(j_{3},m_{3})(j_{5},m_{5})|I_{v}|(j_{2},m_{2})(j_{4},m_{4})\ra$ in the standard spin basis labeled by the spin $j$ and the magnetic moment $m$.}
	\label{fig:intertwiner}
\end{figure}
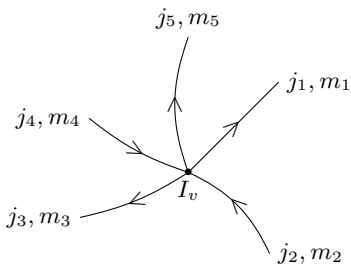
As illustrated on fig.\ref{fig:intertwiner}, an {\it intertwiner} is a $\SU(2)$-invariant  state -or singlet- living in the tensor product of the incoming and outgoing spins at the vertex $v$:
\be
I_{v}\in\textrm{Inv}_{\SU(2)}\Big{[}
\bigotimes_{e|\,v=s(e)} V_{j_{e}}
\otimes
\bigotimes_{e|\,v=t(e)} V_{j_{e}}^{*}
\Big{]}
\,.
\ee
The scalar product between two spin network states based on the same graph $\Gamma$ is then given by the product of the scalar products between their intertwiners:
\be
\la \Psi_{\{j_{e},I_{v}\}}|\Psi_{\{\tilde{j}_{e},\tilde{I}_{v}\}} \ra
=
\la {\{j_{e},I_{v}\}}| {\{\tilde{j}_{e},\tilde{I}_{v}\}}\ra
=
\prod_{e}\delta_{j_{e},\tilde{j}_{e}}
\,\prod_{v}\la I_{v}|\tilde{I}_{v}\ra
\,.
\ee

Loop quantum gravity is formulated on the full Hilbert space of spin network states as a sum over all graphs $\Gamma$ of the Hilbert spaces $\cH_{\Gamma}$ , defined as a projective limit taking into account in a consistent way the inclusion of subgraphs into larger graphs \cite{Ashtekar:1994mh,Thiemann:2007zz}. Then we construct observables as operators either on fixed graphs or allowing transitions between graphs, and we define the dynamics through transition amplitudes, obtained either by suitably regularized Hamiltonian constraint operators \cite{Thiemann:1996aw,Thiemann:1996av,Borissov:1997ji,Gaul:2000ba,Assanioussi:2015gka} or by spinfoam state-sum models inspired from topological field theory \cite{Reisenberger:1996pu,Baez:1997zt,Livine:2010zx,Perez:2012wv}.

In the present work, we are interested in the generalization of the framework to spatial slices with boundaries. As we explain below, such a spatial boundary $\cB=\pp\Sigma$, often referred to as {\it corners} (between space and time) as illustrated on fig.\ref{fig:corner}, is taken into account by extending the definition of spin networks to graphs with open edges.

\subsection{Corners, boundary states and maps}

We consider introducing spin networks on a bounded spatial slice similar to taking a bounded subset of a spin network state. As illustrated on fig.\ref{fig:boundary}, this means considering a graph $\Gamma$ with open edges $e\in\pp\Gamma$ puncturing the boundary surface. We do not endow the boundary with extra structure, representing the 2d boundary intrinsic geometry as in \cite{Freidel:2016bxd,Freidel:2018pvm,Freidel:2019ees} or locality on the boundary as in \cite{Feller:2017ejs}, but discuss the minimal boundary structure.
\begin{figure}[htb]
	\centering
	\begin{tikzpicture} []

\coordinate  (O1) at (0,0);
\coordinate  (O2) at (2,0.7);
\coordinate  (O3) at (4,0);
\coordinate  (O4) at (2,-1);

\coordinate  (P1) at (0.8,3);
\coordinate  (P2) at (2.7,3.7);
\coordinate  (P3) at (5.3,3);
\coordinate  (P4) at (2.6,2);

\draw[dashed,color=blue] (O1) to[out=85,in=160] (O2);
\draw[dashed,color=blue] (O2) to[out=-20,in=95] (O3);
\draw[color=blue]        (O3) to[out=-95,in=2] (O4);
\draw[color=blue]        (O4) to[out=178,in=-80] (O1);

\draw[color=blue] (P1) to[out=90,in=160] (P2);
\draw[color=blue] (P2) to[out=-20,in=95] (P3);
\draw[color=blue]        (P3) to[out=-90,in=2] (P4);
\draw[color=blue]        (P4) to[out=178,in=-80] (P1);

\draw        (O1) to[out=95,in=-100] (P1);
\draw        (O3) to[out=80,in=-87] (P3);

\coordinate  (A1) at (0.4,0);
\coordinate  (A2) at (1.8,0.5);
\coordinate  (A3) at (3.4,0.2);
\coordinate  (A4) at (2.2,-0.6);
\coordinate  (A5) at (1.2,-0.6);

\draw (A1) [color=red] -- ++ (-0.6,-1) node[below]{$e_{1}\in\pp\Gamma$};
\draw (A2) [color=red] -- ++ (-0.6,0.8) node[above]{$e_{5}$};
\draw (A3) [color=red] -- ++ (0.6,0.8)node[above]{$e_{4}$};
\draw (A4) [color=red] -- ++ (0.2,-0.95)node[below]{$e_{3}$};
\draw (A5) [color=red] -- ++ (-0.3,-0.85)node[below]{$e_{2}$};

\draw (A1) [color=green] -- (A2);
\draw (A2) [color=green] -- (A3);
\draw (A3) [color=green] -- (A4);
\draw (A4) [color=green] -- (A5);
\draw (A5) [color=green] -- (A1);
\draw (A2) [color=green] -- (A4);
\draw (A2) [color=green] -- (A5);

\node[scale=0.7,color=red] at (A1) {$\bullet$};
\node[scale=0.7,color=red] at (A2) {$\bullet$};
\node[scale=0.7,color=red] at (A3) {$\bullet$};
\node[scale=0.7,color=red] at (A4) {$\bullet$};
\node[scale=0.7,color=red] at (A5) {$\bullet$};

\coordinate  (B1) at (1.5,3.2);
\coordinate  (B2) at (2.6,3.5);
\coordinate  (B3) at (4.4,3.3);
\coordinate  (B4) at (3.2,2.4);
\coordinate  (B5) at (2.2,2.4);

\draw (B1) [color=red] -- ++ (-0.9,0.8);
\draw (B2) [color=red] -- ++ (-0.3,0.8);
\draw (B3) [color=red] -- ++ (0.6,0.8);
\draw (B4) [color=red] -- ++ (0.2,-0.9);
\draw (B5) [color=red] -- ++ (-0.3,-0.9);

\node[scale=0.7,color=red] at (B1) {$\bullet$};
\node[scale=0.7,color=red] at (B2) {$\bullet$};
\node[scale=0.7,color=red] at (B3) {$\bullet$};
\node[scale=0.7,color=red] at (B4) {$\bullet$};
\node[scale=0.7,color=red] at (B5) {$\bullet$};

\draw (B1) [color=green] -- (B2);
\draw (B2) [color=green] -- (B3);
\draw (B3) [color=green] -- (B4);
\draw (B4) [color=green] -- (B5);
\draw (B5) [color=green] -- (B1);
\draw (B2) [color=green] -- (B4);
\draw (B2) [color=green] -- (B5);

	\end{tikzpicture}
	\caption{On each spatial slice, the boundary states consist in the tensor product of spin states living on the boundary edges of the  spin network: $\cH^{\pp}_{\Gamma}=\bigotimes_{e\in\pp\Gamma}\bigoplus_{j_{e}}\cV_{j_{e}}$.}
	\label{fig:boundary}
\end{figure}
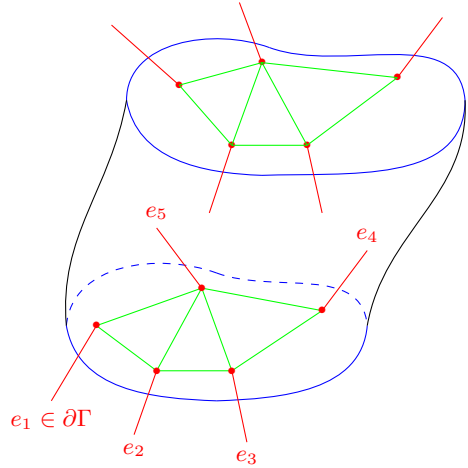

Each boundary edge $e\in\pp\Gamma$ carries a spin $j_{e}$ and a vector in the corresponding representation $v_{e}\in\cV_{j_{e}}$. This defines the boundary Hilbert space as:
\be
\cH^{\{j_{e}\}_{e\in\pp\Gamma}}_{\Gamma}
=
\bigotimes_{e\in\pp\Gamma}\cV_{j_{e}}
\,.
\ee \label{eq:boundary-state}
One does not need to fix the spins carried by the boundary edges and can consider the larger boundary Hilbert space:
\be
\cH^{\pp}_{\Gamma}
=\bigoplus_{\{j_{e}\}}\cH^{\{j_{e}\}_{e\in\pp\Gamma}}_{\Gamma}
=\bigotimes_{e\in\pp\Gamma}\cV_{e}
\qquad\textrm{with}\quad
\cV=\bigoplus_{j}\cV_{j}\,.
\ee
Using the Schwinger realization of the $\su(2)$ Lie algebra in terms of a pair of quantum oscillators, the Hilbert space $\cV$ is the tensor product of two copies of the harmonic oscillator Hilbert space, which can be understood as (holomorphic) wave-functions of a spinor, i.e. a complex 2-vector \cite{Freidel:2009ck,Borja:2010rc,Livine:2011gp,Dupuis:2012vp,Livine:2013zha,Alesci:2015yqa,Bianchi:2015fra,Bianchi:2016tmw}.

Calling $\bGamma=\Gamma\setminus\pp\Gamma$ the bulk or interior of the graph $\Gamma$, a spin network wave-function on the graph $\Gamma$ with boundary is still a function of group elements living on bulk edges $e\in\Gamma\setminus\pp\Gamma$, but is not anymore valued in the field $\C$ but into the boundary Hilbert space $\cH^{\pp}_{\Gamma}$:
\be
\psi(\{g_{e}\}_{e\in\bGamma})\,\in\,\cH^{\pp}_{\Gamma}
\,.
\ee
The scalar product between wave-functions is inherited from the inner product between boundary states:
\be \label{eq:Definition-InnerProduct}
\la \psi|\widetilde{\psi}\ra
=
\int
\prod_{e\in\bGamma}\rd g_{e}\,
\la \psi(\{g_{e}\}_{e\in\bGamma})|\widetilde{\psi}(\{g_{e}\}_{e\in\bGamma})\ra
\,,
\ee
with the normalization of wave-functions reading as:
\be
\la \psi|{\psi}\ra
=
\int
\prod_{e\in\bGamma}\rd g_{e}\,
\la \psi(\{g_{e}\}_{e\in\bGamma})|{\psi}(\{g_{e}\}_{e\in\bGamma})\ra
=1
\,.
\ee

To be more precise, it should actually be considered as a linear form on the boundary Hilbert space and thus live in the dual Hilbert space, $\psi(\{g_{e}\}_{e\in\bGamma})\,\in\,(\cH^{\pp}_{\Gamma})^{*}$.
This means that it defines a distribution on boundary states depending on the group elements, or holonomies, living in the bulk:
\be
\forall \Phi\in \cH^{\pp}_{\Gamma}
\,,\quad
\la \psi(\{g_{e}\}_{e\in\bGamma})\,|\,\Phi\ra \,\,\in\C
\,.
\ee
In simpler words, a spin network state is now a map on boundary states (or corner states), which we will loosely refer to as a boundary map.

\medskip

The statement of gauge invariance also has to take into account the boundary: the wave-function will be invariant with respect to bulk gauge transformations  while it will be covariant under gauge transformations on the boundary.
More precisely, we distinguish bulk vertices $v\in V^{o}$ that are not connected to any boundary edge and boundary vertices $v\in V_{\pp}$ that are attached to at least one boundary edge. The wave-function is assumed to be invariant under $\SU(2)$ transformations acting at bulk vertices, while $\SU(2)$ transformations acting at boundary vertices will act on the spin states dressing the boundary edges:
\be
|\psi(\{h_{t(e)}g_{e}h_{s(e)}^{-1}\})\ra
=
\left(\bigotimes_{e\pp\Gamma} h_{v(e)}^{\eps_{e}^{v}}\right)
\,|\psi(\{g_{e}\})\ra
\,
\ee
where $v(e)$ for $e\in\pp\Gamma$ denotes the vertex to which the boundary edge is attached and $\eps_{e}^{v}=1$ is the boundary edge is outgoing ($v(e)=s(e)$) while $\eps_{e}^{v}=-1$ is the boundary edge is incoming ($v(e)=t(e)$).

\medskip

The definition of the spin network basis states can then be adapted to the case with boundary:
\beq \label{eq:spin-network-with-boundary}
\Psi_{\{j_{e},I_{v}\}}(\{g_{e}\}_{e\in\bGamma})
&=&
\sum_{m_{e}^{t,s}}
\prod_{e\in\bGamma}\sqrt{2j_{e}+1}\,\la j_{e}m_{e}^{t}|g_{e}|j_{e}m_{e}^{s}\ra
\,\prod_{v} \la \bigotimes_{e\in\bGamma|\,v=s(e)} j_{e}m_{e}^{s}
|\,I_{v}\,|
\bigotimes_{e\in\bGamma|\,v=t(e)} j_{e}m_{e}^{t}\ra
\\
&\in&
\bigotimes_{\substack{e\in\pp\Gamma\\ t(e)\in\Gamma}} \cV_{j_{e}}^*
\,\otimes\,
\bigotimes_{\substack{e\in\pp\Gamma\\ s(e)\in\Gamma}} \cV_{j_{e}}
\nn
\,.
\eeq
We sum over the magnetic indices $m$'s only for the bulk edges. The spin states on the boundary edges are not contracted, so that the wave-function $\Psi_{\{j_{e},I_{v}\}}$ is valued in the boundary Hilbert space $\cH^\pp_{\Gamma}$.
This can be made more explicit  by writing the wave-function $\psi$ as a tensor by evaluating on a basis of boundary states,
\be
\psi^{\{j_{e},m_{e}\}_{e\in\pp\Gamma}}
=
\la \otimes_{e\in\pp\Gamma}j_{e},m_{e}\,|\,\psi\ra
\,.
\ee
Assuming that boundary edges are outgoing for the sake of simplicity, this gives for spin network basis states:
\beq
\Psi_{\{j_{e},I_{v}\}}(\{g_{e}\})^{\{j_{e},m_{e}^{s}\}_{e\in\pp\Gamma}}
&=&
\la \otimes_{e\in\pp\Gamma}j_{e},m_{e}^{s}\,|\,\Psi_{\{j_{e},I_{v}\}}(\{g_{e}\})\ra
\\
&=&
\sum_{m_{e}^{t,s}}
\prod_{e\in\bGamma}\sqrt{2j_{e}+1}\,\la j_{e}m_{e}^{t}|g_{e}|j_{e}m_{e}^{s}\ra
\,\prod_{v} \la \bigotimes_{e\in\Gamma|\,v=s(e)} j_{e}m_{e}^{s}
|\,I_{v}\,|
\bigotimes_{e\in\bGamma|\,v=t(e)} j_{e}m_{e}^{t}\ra
\nn
\eeq
The scalar product between those wave-functions is given by the scalar product of the bulk intertwiner as for the no-boundary case:
\beq
\la \Psi_{\{j_{e},I_{v}\}}|\Psi_{\{\tilde{j}_{e},\tilde{I}_{v}\}} \ra
&=&
\sum_{\{k_{e},m_{e}\}}
\overline{\Psi_{\{j_{e},I_{v}\}}(\{g_{e}\})^{\{k_{e},m_{e}\}}}
\,\Psi_{\{\tilde{j}_{e},\tilde{I}_{v}\}}(\{g_{e}\})^{\{k_{e},m_{e}\}_{e\in\pp\Gamma}}
\nn\\
&=&
\prod_{e}\delta_{j_{e},\tilde{j}_{e}}
\,\prod_{v}\la I_{v}|\tilde{I}_{v}\ra
\,.
\eeq

\subsection{Bulk probability}

Now that the bulk wave-function has been promoted to a map from bulk degrees of freedom to boundary state, in a logic following (Atiyah's axiomatization of) topological field theories, the corresponding probability distribution for the bulk fields is given by the boundary space scalar product instead of the mere squared modulus:
\be
\psi(\{g_{e}\}_{e\in\bGamma})\,\in\,\cH^{\pp}_{\Gamma}
\,,
\qquad
\cP(\{g_{e}\}_{e\in\bGamma})=
\la \psi(\{g_{e}\}_{e\in\bGamma})|{\psi}(\{g_{e}\}_{e\in\bGamma})\ra
\,.
\ee
\begin{figure}[htb]
\centering
\begin{tikzpicture}[scale=0.7]
\coordinate (O) at (-5,0);
\path (O) ++(160:2) coordinate (O1);
\path (O) ++(120:2) coordinate (O2);
\path (O) ++(80:2) coordinate (O3);
\path (O) ++(40:2) coordinate (O4);

\draw[thick] (O1) to[bend right=30] (O4);
\draw[thick] (O1) to[bend right=30] (O3);

\draw[thick,red] (O1) -- ++(160:1) ++(160:0.35) node {$j_1$};
\draw[thick,red] (O2) -- ++(120:1) ++(120:0.35) node {$j_2$};
\draw[thick,red] (O3) -- ++(80:1) ++(80:0.35) node {$j_3$};
\draw[thick,red] (O4) -- ++(40:1) ++(40:0.35) node {$j_4$};

   \draw [green,thick,domain=40:160] plot ({-5+2*cos(\x)}, {2*sin(\x)});

   \draw [thick,domain=160:400] plot ({-5+2*cos(\x)}, {2*sin(\x)});
   
   \draw[->,>=stealth,very thick] (-2,0) -- node [midway, above] {gluing its copy} (2,0);
   
\draw (O1) node[scale=0.7,red] {$\bullet$};
\draw (O2) node[scale=0.7,red] {$\bullet$};
\draw (O3) node[scale=0.7,red] {$\bullet$};
\draw (O4) node[scale=0.7,red] {$\bullet$};

\coordinate (A) at (5,0);
\coordinate (B) at (11,0);
\path (A) ++(60:2) coordinate (A1);
\path (A) ++(20:2) coordinate (A2);
\path (A) ++(-20:2) coordinate (A3);
\path (A) ++(-60:2) coordinate (A4);

\path (B) ++(120:2) coordinate (B1);
\path (B) ++(160:2) coordinate (B2);
\path (B) ++(200:2) coordinate (B3);
\path (B) ++(240:2) coordinate (B4);

   \draw [green,thick,domain=120:240] plot ({11+2*cos(\x)}, {2*sin(\x)});
   \draw [green,thick,domain=-60:60] plot ({5+2*cos(\x)}, {2*sin(\x)});

   \draw [thick,domain=240:480] plot ({11+2*cos(\x)}, {2*sin(\x)});
   \draw [thick,domain=60:300] plot ({5+2*cos(\x)}, {2*sin(\x)});

\draw[red,thick] (A1) -- node[above,midway,scale=0.7] {$j_1$} (B1);
\draw[red,thick] (A2) -- node[above,midway,scale=0.7] {$j_2$} (B2);
\draw[red,thick] (A3) -- node[above,midway,scale=0.7] {$j_3$} (B3);
\draw[red,thick] (A4) -- node[above,midway,scale=0.7] {$j_4$} (B4);

\draw[thick] (B1) to[bend left=30] (B4);
\draw[thick] (B1) to[bend left=30] (B3);
\draw[thick] (A1) to[bend right=30] (A4);
\draw[thick] (A1) to[bend right=30] (A3);

\draw (A1) node[scale=0.7,red] {$\bullet$};
\draw (A2) node[scale=0.7,red] {$\bullet$};
\draw (A3) node[scale=0.7,red] {$\bullet$};
\draw (A4) node[scale=0.7,red] {$\bullet$};

\draw (B1) node[scale=0.7,red] {$\bullet$};
\draw (B2) node[scale=0.7,red] {$\bullet$};
\draw (B3) node[scale=0.7,red] {$\bullet$};
\draw (B4) node[scale=0.7,red] {$\bullet$};

\end{tikzpicture}
\caption{
Gluing the two copies of the spin network into the boundary density matrix: boundary edges (red lines) are glued together using the boundary space scalar product, and for each copy; the maximal tree for the bulk gauge fixing consist in the green edges, while the remaining edges, in black, define the independent loops of the bulk graph .
} \label{fig:GluingBoundaryEdges}
\end{figure}
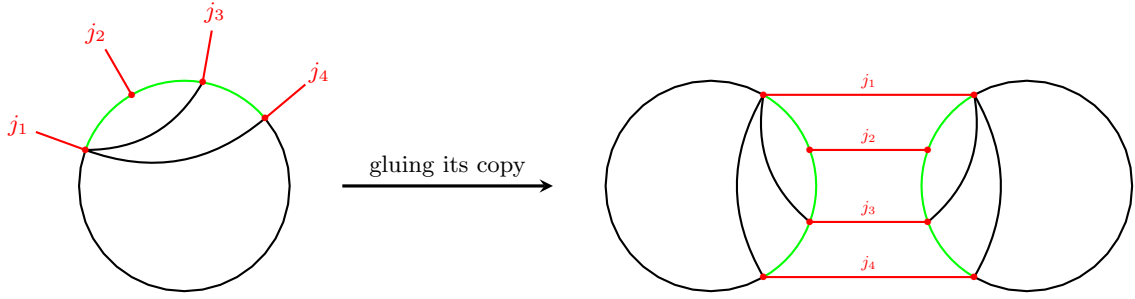

As illustrated on fig.\ref{fig:GluingBoundaryEdges}, we are gluing two copies of the spin network with trivial holonomies along the open edges on the boundary.
This yields a totally gauge-invariant probability distribution, despite the gauge covariance of the waver-function under boundary gauge transformations:
\be
\cP(\{g_{e}\}_{e\in\Gamma})
=
\cP(\{h_{t(e)}g_{e}h_{s(e)}^{-1}\}_{e\in\Gamma})\,\quad
\forall h_{v}\in\SU(2)^{V}
\,,
\ee
with no difference between bulk and boundary vertices or edges.

Following the earlier work on spin networks \cite{Freidel:2002xb} and subsequent works \cite{Livine:2006xk,Livine:2007sy,Livine:2008iq,Livine:2013gna, Charles:2016xwc, Anza:2016fix}, we can gauge-fix  this gauge invariance down to a single $\SU(2)$ action.
To this purpose, one chooses an arbitrary  root vertex $v_{0}\in\Gamma$ and a maximal tree in the bulk graph $T\subset\Gamma^{o}$. A tree is a set of edges that never form any cycle (or loop). A maximal tree $T$ has $(V-1)$ edges. Furthermore, for any vertex $v\in\Gamma$, it defines a unique path of edges $P[v_{0}\rightarrow v]\subset T$ along the tree linking the root vertex $v_{0}$ to the vertex $v$. This allows to gauge-fix all the group elements along tree edges to the identity,  $g_{e\in T}\mapsto\id$, by choosing gauge transformations $h_{v}$ at every vertex but the root vertex as:
\be
h_{v}=\left(\overleftarrow{\prod_{\ell\in P[v_{0}\rightarrow v]}} g_{\ell}\right)^{-1}\,,
\ee
where the product of group elements is taken from right to left over $g_{\ell}$ if the edge $\ell$ is oriented in the same direction than the path $P[v_{0}\rightarrow v]$ and over its inverse $g_{\ell}^{-1}$ otherwise.
This maps all the group elements on tree edges to the identity, $h_{t(e)}g_{e}h_{s(e)}^{-1}=\id$ for $e\in T$. The remaining edges, which do not belong the tree actually correspond to a minimal generating set of loops (or cycles) on the bulk graph $\Gamma^{o}$. Indeed, each non-tree edge defines a loop from the root vertex to the edge and back,
\be
\cL_{e\notin T}:v_{0}\underset{T}{\rightarrow}s(e)\underset{e}{\rightarrow}t(e)\underset{T}{\rightarrow}v_{0}
\,.\nn
\ee
There are $L=E-V+1$  edges not belonging to $T$, defining $L$ such loops. One can show that every cycle on the bulk graph $\Gamma^{o}$ can generated from those cycles.
For $e\notin T$, the gauge transformation built above does not map the group element $g_{e}$ to the identity anymore but maps it to the holonomy around the corresponding loop,
\be
\forall e\notin T\,,\qquad
h_{t(e)}g_{e}h_{s(e)}^{-1}
=
\overleftarrow{\prod_{\ell\in \cL_{e}}} g_{\ell}
\equiv
G_{e}
\,.\nn
\ee
As a consequence, the bulk probability distribution depends only on those $L$ group elements:
\be
\cP(\{g_{e}\}_{e\in\Gamma})
=
\cP(\{G_{e}\}_{e\notin T},\{\id\}_{e\in T})
\equiv
\cP_{GF}(G_{1},..,G_{L})
\,.
\ee
Putting aside the gauge-fixed group elements living on the tree edges and focusing on the non-trivial loop holonomies, this gauge-fixed  bulk probability $\cP_{GF}$ is still invariant under gauge transformation at the root vertex $v_{0}$:
\be
\cP_{GF}(G_{1},..,G_{L})=
\cP_{GF}(h \, G_1 \, h^{-1},\cdots,h \, G_L \, h^{-1})
\,, \quad \forall \, h \in \SU(2)
\,.
\ee
This directly implies two simple results:
\begin{prop} \label{theorem:ExtremalPoint}
The configuration $G_1=\cdots=G_L=\id$, representing a flat $\SU(2)$ connection, is always a stationary point for the bulk probability function $\cP(\{g_{e}\})=\la {\psi}(\{g_{e}\}_{e\in\bGamma}) | {\psi}(\{g_{e}\}_{e\in\bGamma})\ra$.
\end{prop}
\begin{prop} \label{coro:Norm_tree-v.s.-loop}
If the bulk graph $\bGamma$ is a tree, i.e. does not contain any loop, then the bulk probability function $\cP(\{g_{e}\})$ is constant and does not depend on the bulk holonomies $g_{e}$.
\end{prop}

\subsection{Spin network maps as quantum circuits}

We would like to build on the interpretation of spin network wave-functions as valued in the space linear forms on the boundary Hilbert space, or boundary maps. This can be translated operationally as spin networks defining quantum circuits on the boundary data.

Let us fix the spins on the boundary edges and distinguish their orientation. Then a spin network wave-function for the bulk graph defines a family of maps, between the spins on the incoming boundary edges to the spins on the outgoing boundary edges, labeled by the holonomies living on the bulk links:
\be
\psi(\{g_{e}\}_{e\in\bGamma})
\,:\,\,
\bigotimes_{\substack{e\in\pp\Gamma\\ t(e)\in\Gamma}} \cV_{j_{e}}
\longmapsto
\bigotimes_{\substack{e\in\pp\Gamma\\ s(e)\in\Gamma}} \cV_{j_{e}}
\,.
\ee
Of course, we could unfix the boundary spins and more generally attach the larger Hilbert space $\cV=\bigoplus_{j}\cV_{j}$ to each boundary edge.
As illustrated on fig.\ref{fig:spinnetcircuit}, the spin network graph, with its link and node structure, already carries the natural structure of a circuit. The holonomies, or $\SU(2)$ group elements, on the graph links are interpreted as unitary one-spin gates, while the intertwiners, or $\SU(2)$-invariant maps, naturally define multi-spins gates.
\begin{figure}[htb]
	\centering
\begin{tikzpicture} []
\draw[thick,decorate,decoration={brace},xshift=-4pt,yshift=0pt]
(0,0) -- (0,1) -- (0,2) node [black,yshift=-1cm,xshift=-1.35cm] {Incoming edges};
    \draw[thick] (0,1)--node[midway,above]{\footnotesize Spin $| j_2, m_2 \ra$}(2,1) node[midway,sloped]{$>$};
    \draw[thick] (0,2)--node[midway,above]{\footnotesize Spin $| j_1, m_2 \ra$}(2,2) node[midway,sloped]{$>$};
    \draw[thick] (0,0)--node[left,above]{\footnotesize Spin $| j_3, m_3 \ra$} node[midway,sloped]{$>$} (2,0);
    \draw[thick] (3,0) -- (4.5,0) node[midway,sloped]{$>$};
    \draw[thick] (2,-0.5) rectangle (3,2.5) node [pos=.5]{$\cI_A$};
    \draw[thick] (4.5,-0.5) rectangle (5.5,1.5) node [pos=.5]{$\cI_B$};
    \draw[thick] (7,0.5) rectangle (8,2.5) node [pos=.5]{$\cI_C$};
    \draw[thick] (4.5,1)--(4,1);
    \draw (3.75,1) circle (0.25) node {$g_2$};
    \draw[thick] (3.5,1)--(3,1)node[midway,sloped]{$<$};
    \draw[thick] (6.5,1)--(7,1) node[midway,sloped]{$>$};
    \draw (6.25,1) circle (0.25) node {$g_3$};
    \draw[thick] (5.5,1)--(6,1);
    \draw[thick] (3,0)--(4.5,0);
    \draw[thick] (5.5,0)--(8.5,0)node[midway,sloped]{$>$};
    \draw[thick] (3,2)--(4.75,2);
    \draw (5,2) circle (0.25) node {$g_1$};
    \draw[thick] (5.25,2)--(7,2)node[midway,sloped]{$>$};
    \draw[thick] (8,2)--(8.5,2)node[midway,sloped]{$>$};
    \draw[thick,decorate,decoration={brace,mirror},xshift=4pt,yshift=0pt]
    (8.5,0) -- (8.5,2) node [black,midway,xshift=1.35cm] {Outgoing edges};
\end{tikzpicture}
	\caption{Spin network as a quantum circuit: holonomies become unitary one-spin gates while intertwiners are multi-spin gates; the circuit can contains loops.}
	\label{fig:spinnetcircuit}
\end{figure}
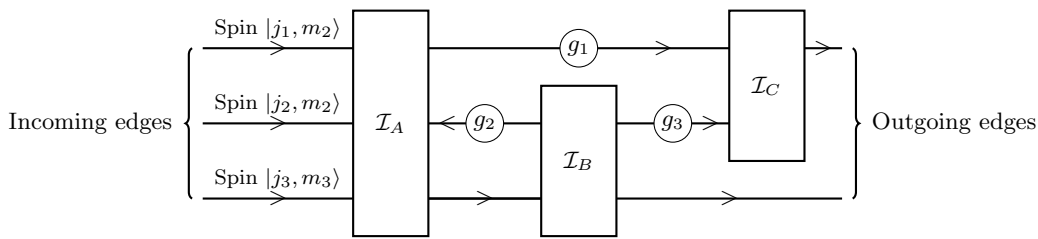

The spin network state is not a process in itself. There are two important points to keep in mind. First, a spin network is a spatial construct, and not directly a space-time structure. A spin network is not a (quantum) causal history (see e.g. \cite{Markopoulou:1999cz,Hawkins:2003vc} for a presentation and discussion on quantum causal histories). The maps that it defines between the boundary spins  are thus possible processes that might occur if the spin network state itself (i.e. the quantum state of 3D geometry) does not evolve. In that sense, it is truly a circuit, to which we haven't yet sent an input and on which we can still adjust some parameters. Indeed, the second important remark is that the holonomies are not fixed. The spin network defines a whole family of boundary maps, which vary in the individual one-spin gates defined by the holonomies $\{g_{e}\}_{e\in\bGamma}$ along the bulk edges. From the point of view of the boundary, these holonomies are not  fixed, they should either be averaged over or some other criteria should be found to determine them. For instance, the holonomies, or more precisely their quantum probability distribution, should ultimately be determined by the dynamics of quantum gravity.
Nevertheless, even without exploring the issue of defining the dynamics of loop quantum gravity, either by a Hamiltonian constraint operator or by spinfoam transition amplitudes, this quantum circuit perspective allows to formulate interesting questions:
\begin{itemize}

\item Working with a given spin network state, with a fixed graph, fixed spins and intertwiners, can we characterize the resulting subset of boundary maps induced by allowing for arbitrary holonomies along the edges? Or vice-versa, how much does a boundary state (for both incoming and outgoing boundary spins) fixes the holonomies in the bulk? Could this be used to formulate a holographic principle for loop quantum gravity?

\item Going further, looking at the spin network state as a black box, with access solely to the boundary spins, if we know the subset of boundary maps that it defines, how much of the bulk graph and intertwiners can we re-construct? Could one think of the diffeomorphism constraints of loop quantum gravity as identifying spin network states which lead to same set of boundary maps? This would be an holographic implementation of the dynamics through bulk-to-boundary coarse-graining, along the lines of \cite{Livine:2017xww}.

\item The issue of defining the dynamics or the coarse-graining of the theory is actually equivalent to the problem of defining a physical inner product or a flow of inner products from the microscopic theory to a coarse-grained macroscopic theory. The quantum circuit perspective offers a possible approach. The microscopic inner product between quantum circuit is defined as the loop quantum gravity kinematical inner product, reflecting the scalar product between intertwiners, i.e. the basic multi-spin gates. As we coarse-grain or sparsify the quantum circuit (while possibly not affecting the boundary maps), we reduce the bulk structure of the circuit by encompassing subsets of holonomies and intertwiners into single larger multi-spin gates, thus leading to a scalar product between those multi-spin gates. The ultimate stage is the fully coarse-grain state, directly provided with the inner product between boundary maps. Studying this in more details would reveal the coarse-graining flow of spin network states in loop quantum gravity.

\end{itemize}

Although these topics are very likely essential to the understanding of the renormalization flow, holographic behavior and semi-classical regime of loop quantum gravity,  they are broad questions out of the scope of the present work and are postponed to future investigation.

\section{Boundary Density Matrix}


\subsection{Bulk state to Boundary density matrix}

We would like to shift the focus from the bulk to the boundary and investigate in more details the boundary state induced by the bulk spin network state defined as the density matrix obtained by integrating over the group elements, or in other words, taking the partial trace over bulk holonomies:
\be \label{eq:Coarse-graining}
\rho_{\pp\Gamma}[\psi]=
\int
\prod_{e\in\bGamma}\rd g_{e}\,
|{\psi}(\{g_{e}\}_{e\in\bGamma})\ra\la \psi(\{g_{e}\}_{e\in\bGamma})|
\,\in
\textrm{End}(\cH^{\pp}_{\Gamma})
\,,
\ee
\be
\tr\,
\rho_{\pp\Gamma}[\psi]
=
\int[\rd g_{e}]\,
\la \psi(\{g_{e}\}_{e\in\bGamma})|{\psi}(\{g_{e}\}_{e\in\bGamma})\ra
=
\int[\rd g_{e}]\,
\cP(\{g_{e}\}_{e\in\bGamma})
\,.\nn
\ee
This mixed state on the boundary can be considered as a coarse-graining of the bulk spin network state \cite{Livine:2006xk,Bianchi:2013toa}.
The goal of this paper is to  compare the data encoded in the bulk wave-function $\psi_{\Gamma}$ and in the induced boundary density matrix $\rho_{\pp\Gamma}$.
\begin{figure}[htb!]
	\centering
	\begin{tikzpicture} []

\coordinate(O1) at (1,3);
\coordinate(O2) at (2.4,3.3);
\coordinate(O3) at (2.7,3);
\coordinate(O4) at (2.3,2.7);
\coordinate(O5) at (2,3.2);
\draw (O1) -- ++(-1,0.4);
\draw (O1) -- ++(-1,-0.4);
\draw (O2) -- ++(1,0.2);
\draw (O3) -- ++(1,0);
\draw (O4) -- ++(0.5,-0.3);

\draw (O1) to[bend left] (O5);
\draw (O1) to[bend right] (O5);
\draw (O2) to (O5);
\draw (O2) to (O3);
\draw (O3) to (O4);
\draw (O4) to (O5);

\node[scale=0.75] at (O1)  {$\bullet$};
\node[scale=0.75] at (O2)  {$\bullet$};
\node[scale=0.75] at (O3)  {$\bullet$};
\node[scale=0.75] at (O4)  {$\bullet$};
\node[scale=0.75] at (O5)  {$\bullet$};

\coordinate(P1) at (1,1);
\coordinate(P2) at (2.4,1.3);
\coordinate(P3) at (2.7,1);
\coordinate(P4) at (2.3,0.7);
\coordinate(P5) at (2,1.2);
\draw (P1) -- ++(-1,0.4);
\draw (P1) -- ++(-1,-0.4);
\draw (P2) -- ++(1,0.2);
\draw (P3) -- ++(1,0);
\draw (P4) -- ++(0.5,-0.3);

\draw (P1) to[bend left] (P5);
\draw (P1) to[bend right] (P5);
\draw (P2) to (P5);
\draw (P2) to (P3);
\draw (P3) to (P4);
\draw (P4) to (P5);

\node[scale=0.75] at (P1)  {$\bullet$};
\node[scale=0.75] at (P2)  {$\bullet$};
\node[scale=0.75] at (P3)  {$\bullet$};
\node[scale=0.75] at (P4)  {$\bullet$};
\node[scale=0.75] at (P5)  {$\bullet$};

\draw[->,>=stealth,very thick] (4,2) -- node [midway, above] {$\displaystyle{ \int \prod_{ e\in \bGamma } \rd g_e }$} (6,2);

\coordinate(A1) at (8,3);
\coordinate(A2) at (8.1,3);
\coordinate(A3) at (9.5,3);
\coordinate(A4) at (9.4,3);
\draw (A1) -- ++(-0.7,0.4);
\draw (A1) -- ++(-0.7,-0.4);
\draw (A3) -- ++(0.7,0);
\draw (A3) -- ++(0.7,0.3);
\draw (A3) -- ++(0.5,-0.4);

\draw (A1) to (A2);
\draw (A3) to (A4);
\draw[dashed] (A2) to[bend left] (A4);

\coordinate(B1) at (8,1);
\coordinate(B2) at (8.1,1);
\coordinate(B3) at (9.5,1);
\coordinate(B4) at (9.4,1);
\draw (B1) -- ++(-0.7,0.4);
\draw (B1) -- ++(-0.7,-0.4);
\draw (B3) -- ++(0.7,0);
\draw (B3) -- ++(0.7,0.3);
\draw (B3) -- ++(0.5,-0.4);

\draw (B1) to (B2);
\draw (B3) to (B4);
\draw[dashed] (B2) to[bend right] (B4);

\draw[color=red] (A2) to[bend left] (B2);
\draw[color=red] (A4) to[bend right] (B4);

\node[scale=0.75] at (A2)  {$\bullet$};
\node[scale=0.75] at (A4)  {$\bullet$};
\node[scale=0.75] at (B2)  {$\bullet$};
\node[scale=0.75] at (B4)  {$\bullet$};


	\end{tikzpicture}
	\caption{Boundary density matrix for spin network basis states. The two copies of the spin networks are the bra $\la \psi |$ and ket $| \psi \ra$ which are glued together by the Haar integration over the bulk holonomies  $\int \prod \rd g_{ e\in \bGamma }$. }
	\label{fig:densitymatrix}
\end{figure}
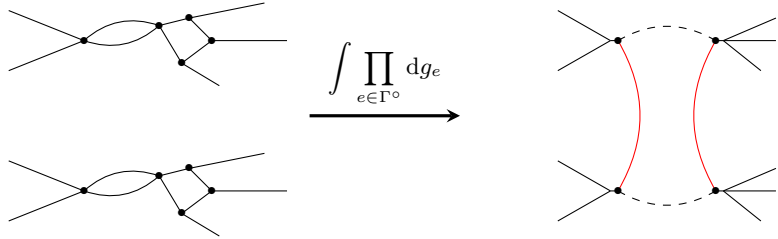

\smallskip

Let us start by looking at normalized pure spin network basis states, i.e. with fixed spins $j_{e}$ and fixed normalized intertwiners $I_{v}$. They are factorized states in the sense that the intertwiners are decoupled so that there is no intertwiner entanglement as discussed in \cite{Livine:2017fgq}. As a result, the boundary state only depends on the intertwiners living on the boundary vertices (i.e. the vertices with at least one boundary edge) and not  on the bulk intertwiners.
Let us insist that ``boundary vertices'' are still in the bulk, the adjective ``boundary'' refers to the fact that they are connected to boundary edges.
Indeed, the orthonormality of the Wigner matrices implies that that each bulk edge is cut in half and both half-edges are glued with their counterparts on the second copy of the wave-function, as illustrated on fig.\ref{fig:densitymatrix}. We get the norm of every bulk intertwiner, normalized to 1, times the contribution from boundary intertwiners which gives the boundary density matrix:
\beq
\la \{\tk_{e},\tm_e\}\,|
\rho_{\pp\Gamma}[\Psi_{\{j_{e},I_{v}\}}]
| \{k_{e},m_e\}\ra
=&\prod_{e}\delta_{k_{e},j_{e}}\delta_{\tk_{e},j_{e}}
\prod_{v\in\pp\Gamma}&
\la
\bigotimes_{\substack{e\in\pp\Gamma\\ v\in e}} j_{e}\tm_{e}
\otimes
\bigotimes_{\substack{e\in\bGamma\\ v=s(e)}} j_{e}m_{e}^{s}
|\,I_{v}\,|
\bigotimes_{\substack{e\in\bGamma\\ v=t(e)}} j_{e}m_{e}^{t}
\ra \nn\\
&&\overline{
\la
\bigotimes_{\substack{e\in\pp\Gamma\\ v\in e}} j_{e}m_{e}
\otimes
\bigotimes_{\substack{e\in\bGamma\\ v=s(e)}} j_{e}m_{e}^{s}
|\,I_{v}\,|
\bigotimes_{\substack{e\in\bGamma\\ v=t(e)}} j_{e}m_{e}^{t}
\ra
}
\,.
\eeq
Assuming that each boundary edge is attached to a different vertex, i.e. each boundary vertex connects to a single boundary edge, this tremendously simplifies. Indeed, as illustrated on fig. \ref{fig:boundaryvertex}, the self-gluing of an intertwiner on itself leads to the identity matrix on the open edge. As a consequence, the density matrix is the totally mixed state with fixed spin on each boundary edge:
\be
\rho_{\pp\Gamma}[\Psi_{\{j_{e},I_{v}\}}]
=
\bigotimes_{e\in\pp\Gamma} \f{\id_{j_{e}}}{(2j_{e}+1)}
\,.
\ee 
This boundary density matrix, for a spin network basis state, clearly does not allow to see the bulk structure!

In the slightly more general case of boundary vertices connected to several boundary edges, the boundary density matrix reflects the first layer of the bulk and ``sees'' the total recoupled spin of the boundary edges for each boundary edge. We will analyze this case in more details in the later section \ref{sec:manyedges}.

\begin{figure}[htb]
\vspace*{5mm}
\begin{subfigure}{0.4\linewidth}
	\begin{tikzpicture}
\coordinate (O1) at (-0.8,0);
\coordinate (O2) at (0.45,0);
\coordinate (O3) at (1.2,0.75);
\coordinate (O4) at (1.2,0);
\coordinate (O5) at (1.2,-0.75);
\draw[thick] (O1) -- node[midway] {$>$} node[midway,above=2.3] {$e\in\pp\Gamma$} node[midway,below=2.3] {$j_e,m_e$} (O2);
\draw (O2) node[scale=0.7] {$\bullet$};
\draw[thick,in=+180,out=+90,scale=3,rotate=0] (O2) to (O3);
\draw[thick] (O2) -- (O4);
\draw[thick,in=+180,out=-90,scale=3,rotate=0] (O2) to (O5);

\coordinate (O6) at (2.25,0.75);
\coordinate (O7) at (2.25,0);
\coordinate (O8) at (2.25,-0.75);
\coordinate (O9) at (3,0);
\coordinate (O10) at (4.25,0);
\draw[thick] (O10) -- node[midway] {$<$} 
node[midway,below=2.3] {$j_e,\tm_e$} (O9);
\draw (O2) node[scale=0.7] {$\bullet$};
\draw (O9) node[scale=0.7] {$\bullet$};
\draw[thick,in=+90,out=+0,scale=3,rotate=0] (O6) to (O9);
\draw[thick] (O7) -- (O9);
\draw[thick,in=-90,out=+0,scale=3,rotate=0] (O8) to (O9);

\draw[dashed] (O3) -- (O6);
\draw[dashed] (O4) -- (O7);
\draw[dashed] (O5) -- (O8);

\node at (-1,-1.5) {$=$};

\coordinate (O11) at (-0.5,-1.5);

\draw[thick] (O11) -- node[midway] {$>$} node[midway,below=2.3] {$j_e,m_e$}  ++ (1.25,0) node[scale=0.7] {$\bullet$} -- node[midway] {$<$} node[midway,below=2.3] {$j_e,\tm_e$}  ++ (1.25,0) ;

	\end{tikzpicture}
\end{subfigure}
\hspace{1cm}
\begin{subfigure}[h]{0.4\linewidth}
\begin{tikzpicture}
\coordinate (A1) at (0,0);
\coordinate (A2) at (0.75,0.5);
\coordinate (A3) at (0.75,0);
\coordinate (A4) at (0.75,-0.5);

\draw[thick,in=+180,out=+60,scale=3,rotate=0] (A1) to (A2);
\draw[thick] (A1) -- (A3);
\draw[thick,in=+180,out=-60,scale=3,rotate=0] (A1) to (A4);

\draw[thick] (A1) -- ++ (135:1);
\draw[thick] (A1) -- ++ (165:1);
\draw[thick] (A1) -- ++ (195:1);
\draw[thick] (A1) -- ++ (225:1);

\coordinate (B1) at (2.55,0);
\coordinate (B2) at (1.8,0.5);
\coordinate (B3) at (1.8,0);
\coordinate (B4) at (1.8,-0.5);

\draw[thick] (B1) -- ++ (45:1);
\draw[thick] (B1) -- ++ (15:1);
\draw[thick] (B1) -- ++ (-15:1);
\draw[thick] (B1) -- ++ (-45:1);

\draw[thick,in=0,out=120,scale=3,rotate=0] (B1) to (B2);
\draw[thick] (B1) -- (B3);
\draw[thick,in=0,out=240,scale=3,rotate=0] (B1) to (B4);

\draw[dashed] (A2) -- (B2);
\draw[dashed] (A3) -- (B3);
\draw[dashed] (A4) -- (B4);

\draw (A1) node[scale=0.7] {$\bullet$} node[above=5] {$v$};
\draw (B1) node[scale=0.7] {$\bullet$} node[above=5] {$v$};

\node at (-0.8,-2) {$ \propto \; \displaystyle{ \sum_{J} \, C_{I_0}[J] }$};

\coordinate (C) at (-0.5,-2);
\coordinate (D) at (1.5,-2);
\coordinate (E) at (2.5,-2);

\draw[thick] (D) -- node[midway,above] {$J$} (E) ;

\draw[thick] (D) -- ++ (135:1);
\draw[thick] (D) -- ++ (165:1);
\draw[thick] (D) -- ++ (195:1);
\draw[thick] (D) -- ++ (225:1);

\draw[thick] (E) -- ++ (45:1);
\draw[thick] (E) -- ++ (15:1);
\draw[thick] (E) -- ++ (-15:1);
\draw[thick] (E) -- ++ (-45:1);

\draw (D) node[scale=0.7] {$\bullet$};
\draw (E) node[scale=0.7] {$\bullet$};

\end{tikzpicture}
\end{subfigure}

	\caption{
	Boundary vertex contribution to the boundary density matrix from the self-gluing of intertwiners:  single boundary edge vs many boundary edges.}
	\label{fig:boundaryvertex}
\end{figure}
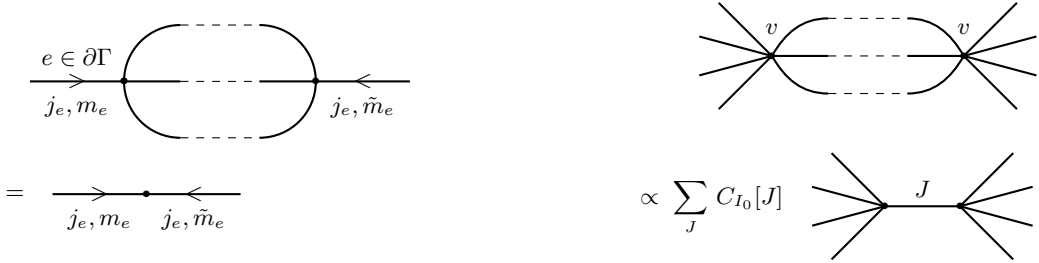

Spin network basis states are actually very peculiar and are a very special case for the bulk quantum geometry. They are eigenstates for geometrical observables, such as areas and volumes, but there are not coherent states with minimal spread on both connection and triad (i.e. on parallel transport and metric) and they do not commute with the Hamiltonian constraints. More generally, physically relevant states will be superposition of such spin network basis states, thus superposition of spins and intertwiners, leading to correlation and entanglement between bulk vertices, in which case  the boundary density matrix will become non-trivial.
Before analyzing in more details the structure of the boundary density matrix, let us underline the main two features of the boundary state as compared to the bulk state:
\begin{itemize}

\item The boundary state $\rho_{\pp\Gamma}$ is typically mixed  even if the bulk spin network state is pure \cite{Livine:2006xk,Bianchi:2013toa}.
Thus a coarse-graining procedure trading the bulk states for the boundary state irremediably creates entropy. In particular, endowing the bulk states with a unitary dynamics would naturally lead to a decoherence process (and possibly re-coherence) for the  boundary states \cite{Feller:2016zuk,Feller:2017ejs}.

\item The boundary state $\rho_{\pp\Gamma}$ does not decompose onto  intertwiners between the boundary spins, even though the bulk spin network is made out of individual intertwiners, as pointed out in \cite{Livine:2006xk,Livine:2013gna,Livine:2017xww}.
Indeed, the density matrix is invariant under the action by conjugation of the $\SU(2)$  group,
\be
\forall h\in\SU(2)\,,\quad
\la \{\tk_{e},\tm_e\}\,|
\, h^{-1}\rho_{\pp\Gamma}[\psi]
h\,| \,\{k_{e},m_e\}\ra
=
\la \{\tk_{e},\tm_e\}\,|
\rho_{\pp\Gamma}[\psi]
| \{k_{e},m_e\}\ra
\ee
where the $\SU(2)$ transformation acts simultaneously on all the boundary edges.  It is however not invariant under gauge transformations acting on left or right, $\rho_{\pp\Gamma}\mapsto h^{-1}\rho_{\pp\Gamma}$ or $\rho_{\pp\Gamma}\mapsto \rho_{\pp\Gamma}h$.
This means that the total spin of the boundary state does not vanish. In fact the boundary defines an intertwiner between the two copies of the wave-function, the bra $\la \psi(\{g_{e}\}_{e\in\bGamma})|$ and the ket $|{\psi}(\{g_{e}\}_{e\in\bGamma})\ra$, as illustrated on fig.\ref{fig:densitymatrix}. The recoupled spin $J$ between the boundary edges defines the overall channel between the bra and the ket. This total spin of the boundary state is called the {\it closure defect} (since the $\SU(2)$ gauge invariance is enforced by the closure constraint, which is a discretization of the Gauss law of the first order formulation of general relativity) \cite{Livine:2013gna,Livine:2019cvi}. The $J=0$ component is the component with vanishing total boundary spin - in usual jargon, the intertwiner component. It represents the ``closed'' or flat component, while the components with $J\ne 0$ can be interpreted as bulk curvature.
From the viewpoint of coarse-graining, it reflects that curvature builds up when gluing flat blocks -the intertwiners- together \cite{Livine:2013gna}.  The gauge symmetry breaking at the boundary, due to allowing $J\ne 0$, can be also be understood as responsible for the entropy of isolated horizon (and thus black holes) in the loop quantum gravity framework \cite{Donnelly:2008vx,Donnelly:2011hn,Livine:2017fgq} (see \cite{Donnelly:2014gva,Donnelly:2016auv} for a more general discussion of gauge symmetry and symmetry breaking on the boundary of gauge field theories).
At the end of the day, the closure defect, or total spin, provides a very useful basis to study the structure of boundary states and of induced boundary density matrices.

\end{itemize}

\subsection{The closure defect basis and $\SU(2)$-invariance of the boundary density matrix}

We would like to introduce the closure defect basis for boundary states, which amounts to decompose them according to the total boundary spin.
Assuming that the boundary edges are all incoming (or all outgoing) to simplify the orientation conventions, we recouple all the boundary spins $j_{e}$ into their total spin $J$:
\be
\cH^{\pp}_{\Gamma}
=
\bigoplus_{ \{j_e\}_{e\in\pp\Gamma} }\bigotimes_{e} \cV_{j_{e}}
=
\bigoplus_{ \{j_e\}_{e\in\pp\Gamma} }\bigoplus_{J} \cV_J \otimes \cN_J^{\{j_{e}\}}\,,
\ee
where the multiplicity spaces (or degeneracy spaces) $\cN_J^{\{j_{e}\}}$ consist in the spaces of intertwiners (i.e. $\SU(2)$-invariant states) in the tensor product of the total spin Hilbert space $\cV^{J}$ with the individual spins $\bigotimes_{e} \cV_{j_{e}}$,
\be
\cN_J^{\{j_{e}\}}
:=
\textrm{Inv}_{\SU(2)}\left[
\cV_J\otimes\bigotimes_{e\in\pp\Gamma} \cV_{j_{e}}
\right]
\,.
\ee
Here, due to the bulk spin network structure, the total spin $J$ is necessary an integer.
Instead of the decoupled basis $|\{j_{e},m_{e}\}_{e}\ra$, we use the recoupled basis, as illustrated on fig.\ref{fig:recoupledbasis}:
\be
\cH^{\pp}_{\Gamma}
=
\bigoplus_{ \{j_e\}_{e\in\pp\Gamma} }\bigoplus_{J,M}\bigoplus_{I^{(J,\{j_{e}\})}}
\C|J,M\ra\otimes|(J,\{j_{e}\}), I\ra
=
\bigoplus_{J,M}\bigoplus_{ \{j_e\} }\bigoplus_{I^{(J,\{j_{e}\})}}
\C|J,M\ra\otimes|(J,\{j_{e}\}), I\ra
\,,
\ee
where the $I^{(J,\{j_{e}\})}=|(J,\{j_{e}\}), I\ra$'s are a basis of intertwiners  in the multiplicity space $\cN_J^{\{j_{e}\}}$ . We might write $I^{(J)}$ instead of  $I^{(J,\{j_{e}\})}$ whenever we don't need to explicitly specify the value of the boundary spins.
These intertwiner states not only encode the recoupled total spin  $J$, but also how the individual spins $j_{e}$ are weaved together.
\begin{figure}[hbt!]
\centering
\vspace*{5mm}
\begin{tikzpicture}[scale=0.7]

\coordinate (O) at (-4.2,0);

\draw[thick] (O) -- ++ (0.5,0);
\draw[thick] (O) ++ (0,1) -- ++ (0.5,0);
\draw[thick] (O) ++ (0,0.5) -- ++ (0.5,0);
\draw[thick,loosely dotted] (O) ++ (0.25,-0.15) -- ++ (0,-0.9);
\draw[thick] (O) ++ (0,-1) -- ++ (0.5,0);

\coordinate (O1) at (-1,0);

\draw (O1)++(3,0) node {$\sim$};

\draw (O1) node { $ | \{ j_e, m_e \}_{e\in\pp\Gamma} \ra \in\cH_{\pp\Gamma} $};

\coordinate (A1) at (5,0);
\coordinate (A2) at (4,1);
\coordinate (A3) at (4,0.5);
\coordinate (A4) at (4,0);
\coordinate (A5) at (4,-1);

\draw[thick] (A1) -- (A2) -- ++ (-0.5,0);
\draw[thick] (A1) -- (A3) -- ++ (-0.5,0);
\draw[thick] (A1) -- (A4) -- ++ (-0.5,0);
\draw[thick] (A1) -- (A5) -- ++ (-0.5,0); 

\draw[thick,loosely dotted] (3.8,-0.15) -- (3.8,-0.9);

\draw (A1) node[scale=0.7] {$\bullet$} node[above=2] {$I$};

\draw[thick] (A1) -- ++ (1.2,0) node[below] {$J,M$} ++ (3.5,-0.5) node {$\underbrace{ | J,M \ra }_{ 
\overset{ \in }
{   \phantom{ \big( } \cV_J \phantom{\big)} }
 } 
 { \phantom{ \big( } \otimes  }
 \underbrace{ | (J,\{j_e\}),I \ra }_{ 
 \overset{ \in }
 {
 \phantom{ \big( } \textrm{Inv} \left( \cV_J \otimes \bigotimes_{e}\cV_{j_e} \right) \phantom{ \big) }
 }
    } $};

\end{tikzpicture}
%
%
\caption{
Recoupled basis for  boundary states in terms of the total boundary spin (or closure defect) $J$.}
\label{fig:recoupledbasis}
\end{figure}
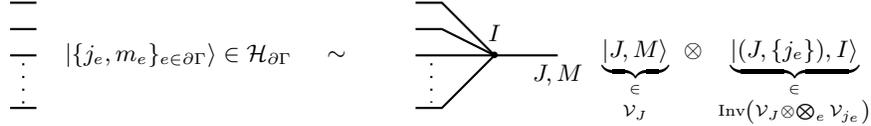

In the framework of the coarse-graining of spin networks introduced in \cite{Charles:2016xwc}, the total spin $J$ is the tag and the multiplicity states $I\in\cN_J^{\{j_{e}\}}$ are tagged intertwiners.
From a physical standpoint, the multiplicity spaces $\cN_J^{\{j_{e}\}}$ for spin recoupling give the black hole horizon micro-states in a na\"ive leading order approach to black hole (micro-canonical) entropy and holography in loop quantum gravity, e.g. \cite{Ashtekar:1997yu,Domagala:2004jt,Livine:2005mw,Agullo:2009eq,Livine:2012cv,Asin:2014gta}.

Let us focus on the case with fixed boundary spins $\{j_{e}\}$, although this is a mere alleviation of the notations, since the spins $j_{e}$ can be implicitly absorbed in the definition of the recoupling intertwiner $I^{(J,\{j_{e}\})}$. The bulk wave-function evaluated on bulk holonomies is a boundary state and can thus be decomposed onto the recoupled basis:
\be \label{eq:Bulk-BoundaryGeneticForm}
| \psi(\{g_{e}\}_{e\in\bGamma}) \ra
=
\sum_{J}\sum_{M}\sum_{I^{(J)}} C_{JMI}(\{g_{e}\}_{e\in\bGamma}) |J,M\ra \otimes |J,I^{(J)}\ra
\,,
\ee
where the coefficients $C_{JMI}(\{g_{e}\}_{e\in\bGamma})$ reflect the internal bulk structure of the wave-functions and depend on the bulk spins and intertwiners.
$\SU(2)$ gauge transformation act non-trivially on the wave-function by the group action on the boundary spins. Now, as we have seen earlier, the density matrix $\rho_{\pp}=\int \rd g_{e}\,|\psi(g_{e})\ra\la \psi(g_{e})|$ is invariant under  conjugation by the simultaneous $\SU(2)$ action on all the boundary spins $\bigotimes_{e} D^{j_{e}}(h)$. This is a direct consequence of the bulk $\SU(2)$ gauge invariance,
\beq
\rho_{\pp\Gamma}[\psi]
&=&
\int \prod_{e} \rd g_{e} \, | \psi(\{ g_{e} \}_{e\in\bGamma}) \ra\la (\{g_{e}\}_{e\in\bGamma} ) |
=
\int \prod_{e} \rd g_{e} \, | \psi(\{ h \, g_{e} \, h^{-1} \}_{e\in\bGamma}) \ra\la (\{h \, g_{e} \, h^{-1}\}_{e\in\bGamma}) |
\\
&=&
\int \prod_{e} \rd g_{e} \, h | \psi(\{ g_{e} \}_{e\in\bGamma}) \ra\la (\{g_{e}\}_{e\in\bGamma}) | h^{-1}
=h \, \rho_{\pp\Gamma}[\psi] \, h^{-1}
\,, \qquad
\forall\, h \in \SU(2)
\,.
\eeq
This $\SU(2)$ action on the boundary boils down to the $\SU(2)$ action on the recoupled spin $D^{J}(h)$ and does not touch the multiplicity sector,
%
%
\be \label{eq:BoundarySU(2)action}
h \triangleright | \psi(\{g_{e}\}_{e\in\bGamma}) \ra
=
\sum_{J}\sum_{M,N}\sum_{I^{(J)}} C_{JMI}(\{g_{e}\}_{e\in\bGamma}) \, D^J_{\tilde{M}M}(h) \, |J,\tilde{M}\ra \otimes |J,I\ra
\,.
\ee
%
%
This means that the invariance of the boundary density matrix, $h\,\rho_{\pp}\,h^{\dagger}=\rho_{\pp}$ for all group elements $h\in\SU(2)$ implies it is necessarily  totally mixed on each subspace at fixed total spin $J$ and that all the information is encoded in the multiplicity subspaces. This is expressed more precisely by the following lemma:

\begin{lemma}
A normalized $\SU(2)$-invariant density matrix $\rho$, thus satisfying $h\,\rho\,h^{\dagger}=\rho\,, \forall\, h\in \SU(2)$, has the following form:
\be \label{eq:SU(2)-invariant}
\rho
=
\bigoplus_{J} p(J) \f{ \id_{\cV_J} }{2J+1} \otimes \rho_{\cN_{J}}
\,, \qquad
\tr \rho_{\cN_{J}}=1\,,\forall J\in\N
\,,\qquad
\tr \rho=\sum_{J}p(J)=1
\,.
\ee
The coefficients $p(J)$ define the probability distribution over the total spin $J$. The operator $\id_{\cV_J}=\sum_{M}| J,M\ra\la J,M|$ is the identity on $\cV_J$ and $\rho_{\cN_{J}}$ is an arbitrary density matrix in the multiplicity space $\cN_{J}$.
\end{lemma}
The $\SU(2)$ invariance  is a key property of the boundary density matrix, which descends directly from the gauge invariance of the bulk wave-functions under local $\SU(2)$ transformations.
Let us stress the important point that this is a statistical invariance under the $\SU(2)$ action, at the level of the density matrix. This does not amount to the invariance of  pure quantum states on the boundary. Indeed strict $\SU(2)$ invariance of the wave-function (i.e. $h\,\rho=\rho\,h^{\dagger}=\rho$) would require $J=0$, while we can have here an arbitrary distribution over all (allowed) values of the total spin $J$.

\subsection{Universal bulk reconstruction from the boundary density matrix}

The natural question is how much can we know about the bulk structure from the boundary density matrix. For instance, does the combinatorial structure of the bulk graph deeply affect the type of boundary density matrix one gets? Here, we  show a universal reconstruction  procedure. As hinted by the work in \cite{Livine:2017xww}, a single bulk loop is enough to get arbitrary boundary density matrices. More precisely, any arbitrary $\SU(2)$-invariant density matrix on the boundary Hilbert space can be induced from a pure bulk state on the single loop bulk graph. We prove this powerful result  below. This can be understood as a boundary-to-bulk purification theorem.

\begin{prop} \label{prop:BoundaryDensityMatrix}
A mixed state $\rho$  on the boundary Hilbert space $\cH_{\pp}$ is  $\SU(2)$-invariant, $h\,\rho\,h^{\dagger}=\rho$, if and only if it is an induced boundary density matrix (IBDM) from a pure (gauge-invariant) bulk state $| \psi(\{g_{e}\}_{e\in\bGamma}) \ra$ for some bulk graph $\bGamma$ connecting the boundary edges.
\end{prop}
\begin{proof}
We already know that induced boundary density matrices are $\SU(2)$-invariant. We have to show the reverse statement. Let us consider an arbitrary $\SU(2)$-invariant density matrix, 
\be
\rho
=
\bigoplus_{J} p(J) \f{ \id_{\cV_J} }{2J+1} \otimes \rho_{\cN_{J}}
\,,\nn
\ee
and let us diagonalize the density matrices for each multiplicity subsector,
\be
\label{eq:SU(2)-inv.DM}
\rho_{\cN_J}=\sum_{r=1}^{R_{J}} W_{I_{r}}^{(J)} \,| J, I_{r}^{(J)} \ra\la J, I_{r}^{(J)} |
\,,
\ee
where $R_{J}$ is the rank of $\rho_{\cN_J}$ and the intertwiners $I_{r}^{(J)}$ are orthonormal states in the multiplicity space $\cN_J$.


Let us consider the bulk graph, as \cite{Livine:2017xww}, with a single vertex tying all the boundary edges to a single loop as drawn on fig \ref{fig:LoopySpinNetwork}.
Then a spin network state is a superposition of intertwiners between the boundary spins and the (pair of) spin(s) carried by the loop. We can unfold this intertwiner with a (virtual) link between the boundary edge and the loop. This (virtual) intermediate link carries the total boundary spin $J$. For each value of $J$, we need to specify the spin $k$ carried by the loop and the two intertwiners at the nodes. The three-valent intertwiner recoupling the loop spin $k$ to the total spin $J$ is unique (when it exists), while the intertwiner recoupling the boundary spins $\{j_{e}\}$ into $J$ will naturally be the intertwiners $I_{r}^{(J)}$.
\begin{figure}[hbt!]
\centering
\begin{tikzpicture}[scale=0.7]

\coordinate (O) at (0,0);
\coordinate (A) at (-6,0);

\draw (A) node {$\rho_{\pp}$};

\draw [domain=0:360] plot ({-6+1.75 * cos(\x)}, {1.75 * sin(\x)});
\draw[thick] (A) ++(-0.6,0);
\draw[thick] (A) ++(0:1) --++ (0:1.5);
\draw[thick] (A) ++(60:1) --++ (60:1.5) ++(60:0.35) node {$j_1$};
\draw[thick] (A) ++(120:1) --++ (120:1.5) ++(120:0.35) node {$j_2$};
\draw[thick] (A) ++(180:1) --++ (180:1.5) ++(180:0.35) node {$j_3$};
\draw[thick] (A) ++(240:1) --++ (240:1.5);
\draw[thick] (A) ++(300:1) --++ (300:1.5);

\draw [thick, loosely dotted,domain=195:230] plot ({-6+2.6 * cos(\x)}, {2.6 * sin(\x)});


\coordinate (O) at (3.5,0);

\draw[->,>=stealth,very thick] (-2,0) -- node[above] {?} (0,0);

\draw[thick,red] (O) -- ++ (315:1.5) node[very near end,above=2] {$J$} coordinate (B) node[blue,scale=0.7] {$\bullet$};
\draw[blue,thick,in=-90,out=0,scale=4.5,rotate=0] (B)  to[loop] node[near start,sloped] {$>$} node[near end,left=2] {$k$} (B) ++(315:0.35) node {$g$};

\draw[thick] (O) -- ++ (0:1.5);
\draw[thick] (O) -- ++ (45:1.5) ++(45:0.35) node {$j_1$};
\draw[thick] (O) -- ++ (90:1.5) ++(90:0.35) node {$j_2$};
\draw[thick] (O) -- ++ (135:1.5) ++(135:0.35) node {$j_3$};
\draw[thick] (O) -- ++ (180:1.5); 
\draw[thick] (O) -- ++ (225:1.5); 

\draw (O) node[scale=0.7] {$\bullet$};

\draw [thick, loosely dotted,domain=160:200] plot ({3.5+1.8 * cos(\x)}, {1.8 * sin(\x)});

\end{tikzpicture}
%
%
\caption{
The universal reconstruction procedure purifying a  $\SU(2)$-invariant boundary density matrix into a pure spin network superposition for a  bulk made of a single vertex and single loop.
}
\label{fig:LoopySpinNetwork}
\end{figure}
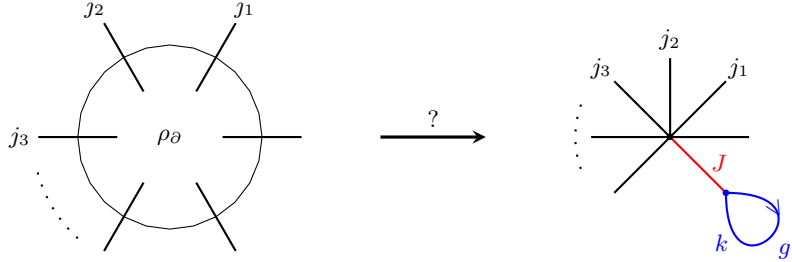

Indeed, for each value of the total spin $J$, we choose $R_{J}$ distinct spins $k_{r}^{(J)}$ for the loop with $J\leq 2k_{r(J)}$, so that $\cV_{J}\subset \cV_{k}\otimes \cV_{k}$, i.e. the loop spin can recouple to $J$, i.e. there exists a 3-valent intertwiner (given by the corresponding Clebsh-Gordan coefficients). We then define the following pure spin network for the $1$-loop graph, in terms of a single bulk holonomy $g$ on the loop,
\be
| \psi(g) \ra
=
\sum_{J,M}  \sqrt{ p(J) } | J,M \ra \otimes \sum_{r}^{R_{J}} \sum_{m,n}^{k_{r}^{(J)}} (-1)^{k_{r}^{(J)}+m} \, \sqrt{ 2k_{r}^{(J)}+1 } \, D^{k_{r}^{(J)}}_{nm}(g) \, \begin{pmatrix}
J & k_{r}^{(J)}& k_{r}^{(J)} \\
M & -m & n
\end{pmatrix} \, \sqrt{ W_{I_{r}}^{(J)} } \, | J, I_{r}^{(J)} \ra
\,.
\ee
It is straightforward to check this boundary pure state leads back to the wanted $\SU(2)$-invariant density matrix \eqref{eq:SU(2)-inv.DM} upon integration over the bulk holonomy $g$.

\end{proof}
It is quite remarkable that the superposition of loop spins and bulk intertwiners naturally leads to  mixed boundary density matrices. 

\subsection{Probing the first layer of the bulk: Bouquets of boundary edges}
\label{sec:manyedges}

Up to now, we have defined the boundary density matrix induced by a bulk spin network state, underlined the fact that the resulting boundary density matrix is typically mixed for a pure spin network state and showed how to construct such a pure bulk state on a graph with at least one loop given a (suitably gauge invariant) boundary density matrix. This universal reconstruction procedure, given above in the  proof of Proposition \ref{prop:BoundaryDensityMatrix}, with a bulk graph made of a single vertex and a single bulk loop, should be understood as a purification of the boundary density matrix into a bulk state. There are nevertheless many possible bulk states on possibly complicated graphs inducing the same boundary state, leading to many ways to purify a given mixed boundary state. In light of this fact, we wish to understand better how the bulk graph structure and potential correlations between the spins and intertwiners within the bulk possibly get reflected in the boundary density matrix.


In this section, we would like to start diving into the bulk, or at least start probing the first layer of the bulk beyond the boundary edges. More precisely, we would like to see the ``boundary vertices'', i.e. the vertices to which are attached boundary edges, and understand if a finer study of the boundary density matrix allows to extract the information about whether bunches of boundary edges are attached to the same boundary vertex.
Indeed, although a rather natural assumption is that each boundary edge to connected to a different vertex in the bulk, this is not a generic configuration. A more general configuration involves boundary edges regrouped into bouquets, each attached to a vertex, as illustrated on fig.\ref{fig:bouquet}.
\begin{figure}[hbt!]
\centering
\begin{tikzpicture}[scale=0.7]
\coordinate (O) at (0,0);
\draw (O) ++(0.6,0);
\coordinate (O1) at (60:1);
\coordinate (O2) at (120:1);
\coordinate (O3) at (240:1);
\coordinate (O4) at (300:1);

   \draw [domain=0:360] plot ({cos(\x)}, {sin(\x)});

\draw (O1) -- ++ (60:0.7) node [right,midway] {$J_1$};
\draw (O2) -- ++ (120:0.7) node [left,midway] {$J_2$};
\draw (O3) -- ++ (240:0.7) node [left,midway] {$J_3$};
\draw (O4) -- ++ (300:0.7) node [right,midway] {$J_4$};

\draw (O1) ++ (60:0.7) -- ++ (0:0.7);
\draw (O1) ++ (60:0.7) -- ++ (45:0.7);
\draw (O1) ++ (60:0.7) node[scale=0.7,blue] {$\bullet$} -- ++ (90:0.7);

\draw (O2) ++ (120:0.7) -- ++ (90:0.7);
\draw (O2) ++ (120:0.7) -- ++ (135:0.7);
\draw (O2) ++ (120:0.7) node[scale=0.7,blue] {$\bullet$} -- ++ (180:0.7);

\draw (O3) ++ (240:0.7) -- ++ (180:0.7);
\draw (O3) ++ (240:0.7) -- ++ (225:0.7);
\draw (O3) ++ (240:0.7) node[scale=0.7,blue] {$\bullet$} -- ++ (270:0.7);

\draw (O4) ++ (300:0.7) -- ++ (270:0.7);
\draw (O4) ++ (300:0.7) -- ++ (315:0.7);
\draw (O4) ++ (300:0.7) node[scale=0.7,blue] {$\bullet$} -- ++ (0:0.7);

\draw[->,>=stealth,very thick] (-4,0) -- (-2,0);
\coordinate (A) at (-6,0);
   \draw [domain=0:360] plot ({-6+cos(\x)}, {sin(\x)});
\draw (A) ++(-0.6,0);
\path (A) ++(60:1) coordinate (A1);
\path (A) ++(120:1) coordinate (A2);
\path (A) ++(240:1) coordinate (A3);
\path (A) ++(300:1) coordinate (A4);

\draw (A1) node[scale=0.7] {$\bullet$};
\draw (A2) node[scale=0.7] {$\bullet$};
\draw (A3) node[scale=0.7] {$\bullet$};
\draw (A4) node[scale=0.7] {$\bullet$};

\draw (A1) -- ++ (0:0.7);
\draw (A1) -- ++ (45:0.7);
\draw (A1) -- ++ (90:0.7);

\draw (A2) -- ++ (90:0.7);
\draw (A2) -- ++ (135:0.7);
\draw (A2) -- ++ (180:0.7);

\draw (A3) -- ++ (180:0.7);
\draw (A3) -- ++ (225:0.7);
\draw (A3) -- ++ (270:0.7);

\draw (A4) -- ++ (270:0.7);
\draw (A4) -- ++ (315:0.7);
\draw (A4) -- ++ (0:0.7);

\draw (O1) node[scale=0.7,blue] {$\bullet$};
\draw (O2) node[scale=0.7,blue] {$\bullet$};
\draw (O3) node[scale=0.7,blue] {$\bullet$};
\draw (O4) node[scale=0.7,blue] {$\bullet$};

\end{tikzpicture}
\caption{
Bouquets of boundary edges attached to boundary vertices $v\in V^{\pp}$ and the chicken feet basis labeled by the recoupled spin $J_{v}$ for each bouquet.}
\label{fig:bouquet}
\end{figure}
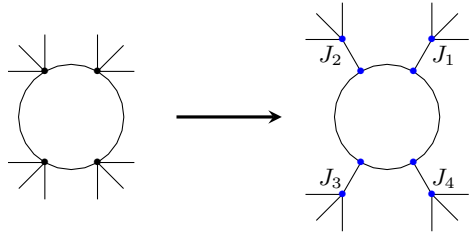

This leads us to introduce a ``chicken feet'' basis where we recouple the spin of the boundary edges of each bouquet separately instead of only considering the total recoupled spin $J$. We thus introduce the bouquet spin $J_{v}$ for each boundary vertex $v$. Writing $V^{\pp}$ for the set of boundary vertices, the boundary Hilbert space decomposes as:
\be
\cH_{\pp}=\bigoplus_{\{j_{e}\}_{e\in\pp}}\bigotimes_{e\in\pp}\cV_{j_{e}}
=\bigoplus_{\{J_{v}\}_{v\in V^{\pp}}} \bigotimes_{v\in V^{\pp}}\cV_{J_{v}}\otimes \cN_{\{J_{v}\}}
\,,
\ee
\be
\cN_{\{J_{v}\}}
=
\bigotimes_{v\in V^{\pp}}\bigoplus_{\{j_{e}\}_{e\,| v\in\pp e}}\textrm{Inv}\Big{[}
\cV_{J_{v}}\otimes \bigotimes_{e\,| v\in\pp e} \cV_{j_{e}}
\Big{]}
\,,
\ee
leading to the chicken feet basis states $|\{J_{v}\}_{v\in V^{\pp}},\{j_{e}\}_{e\in\pp},\{\cI^{J_{v}}_{\{j_{e}\}}\}_{v\in V^{\pp}}\ra$ labelled by the boundary edge spins $j_{e}$, the boundary bouquet spins $J_{v}$ and the intertwiners recoupling them,
as depicted on fig.\ref{fig:bouquet}.

As for the bulk, we similarly unfold the intertwiner states living on the boundary vertices and decompose them into two intertwiners, one ``boundary'' component which recouples all the boundary spins into $J_{v}$ and one ``bulk'' component which recouples the spins on the remaining bulk edges attached to the vertex to $J_{v}$, as illustrated on fig.\ref{fig:boundaryintertwiner}.
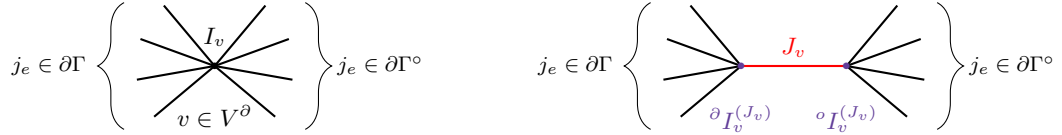
\begin{figure}[hbt!]
\centering
\begin{tikzpicture}[scale=0.7]

\coordinate (O) at (0,0);

\coordinate (A1) at (145:2.1);
\coordinate (A2) at (215:2.1);

\coordinate (B1) at (35:2.1);
\coordinate (B2) at (-35:2.1);

\draw (O) node[scale=0.7] {$\bullet$} ++ (0,0.55) node {$I_v$} ++ (0,-1.5) node{$v\in V^{\pp}$ };

\draw[thick] (O) -- ++ (130:1.5) ;
\draw[thick] (O) -- ++ (160:1.5) ;
\draw[thick] (O) -- ++ (190:1.5) ;
\draw[thick] (O) -- ++ (220:1.5) ;

\draw[thick] (O) -- ++ (50:1.5) ;
\draw[thick] (O) -- ++ (20:1.5) ;
\draw[thick] (O) -- ++ (-10:1.5) ;
\draw[thick] (O) -- ++ (-40:1.5) ;

\draw [decorate,decoration={brace,amplitude=10pt},xshift=-4pt,yshift=0pt]
(A2) -- (A1) node [black,midway,xshift=-1cm] {\footnotesize $j_{e} \in \pp\Gamma$};

\draw [decorate,decoration={brace,amplitude=10pt,mirror},xshift=-4pt,yshift=0pt]
(B2) -- (B1) node [black,midway,xshift=1cm] {\footnotesize $j_{e} \in \pp\bGamma$};

\coordinate (O1) at (10,0);
\coordinate (O2) at (12,0);

\path (O1) ++(145:2.1) coordinate (C1);
\path (O1) ++(215:2.1) coordinate (C2);
\path (O2) ++(35:2.1) coordinate (D1);
\path (O2) ++(-35:2.1) coordinate (D2);

\draw[thick,red] (O1) -- node[above] {$J_v$} (O2);

\draw [decorate,decoration={brace,amplitude=10pt},xshift=-4pt,yshift=0pt]
(C2) -- (C1) node [black,midway,xshift=-1cm] {\footnotesize $j_{e} \in \pp\Gamma$};

\draw [decorate,decoration={brace,amplitude=10pt,mirror},xshift=-4pt,yshift=0pt]
(D2) -- (D1) node [black,midway,xshift=1cm] {\footnotesize $j_{e} \in \pp\bGamma$};

\draw[thick] (O1) -- ++ (130:1.5) ;
\draw[thick] (O1) -- ++ (160:1.5) ;
\draw[thick] (O1) -- ++ (190:1.5) ;
\draw[thick] (O1) node[RoyalPurple,scale=0.7] {$\bullet$} -- ++ (220:1.5) ;

\draw[thick] (O2) -- ++ (50:1.5) ;
\draw[thick] (O2) -- ++ (20:1.5) ;
\draw[thick] (O2) -- ++ (-10:1.5) ;
\draw[thick] (O2) node[RoyalPurple,scale=0.7] {$\bullet$} -- ++ (-40:1.5) ;

\draw[RoyalPurple] (O1) ++ (0,-1) node{$\bI_{v}^{(J_v)} $ };
\draw[RoyalPurple] (O2) ++ (0,-1) node{$\BI_{v}^{(J_v)} $ };

\end{tikzpicture}
\caption{
Unfolding intertwiners on boundary vertices: the decomposition into boundary and bulk intertwiner components.}
\label{fig:boundaryintertwiner}
\end{figure}
Decomposed intertwiner basis states are then labeled by the boundary and bulk spins attached to the (boundary) vertex, the bouquet spin $J_{v}$ and the two boundary and bulk intertwiner, $|\{j_{e}\}_{e\in\pp},J_{v},\bI_{v}^{(J_{v})},\BI_{v}^{(J_{v})}\ra$.

The reconstruction of the first layer of the bulk from the boundary density matrix simply reflects the fact that the boundary component of the intertwiners $I_{v}$ at a vertex attached to some boundary edges matches the boundary intertwiner recoupling the boundary spins to their bouquet spins, i.e. $\cI^{J_{v}}_{\{j_{e}\}}=\bI_{v}^{(J_{v})}$ for a boundary vertex $v\in V^{\pp}$ and for all values of the bouquet spin $J_{v}$.
Let us see more precisely how this gets encoded   into the boundary density matrix.

\medskip

To alleviate the notations, let us fix the  spins $j_{e}$ on the boundary edges $e\in\pp\Gamma$, although it is straightforward to allow arbitrary superpositions of the boundary spins.
In light of the $\SU(2)$ gauge transformations at the vertices and the resulting $\SU(2)$ gauge invariance of the boundary density matrix at each boundary vertex, a boundary density matrix necessarily reads:
\be
\rho_{\pp}
=
\sum_{\{J_{v}\}_{v\in V^{\pp}}}
\bigotimes_{v\in V^{\pp}}\f{\id_{J_{v}}}{(2J_{v}+1)} \otimes \rho_{\{J_{v}\}}
\,,
\ee
where, for each value of the bouquet spins $\{J_{v}\}$, we have the totally mixed state on the spin states and a possibly non-trivial density matrix  $\rho_{\{J_{v}\}}$ on the corresponding multiplicity space,
\be
 \rho_{\{J_{v}\}}\in\textrm{End}[\cN_{\{J_{v}\}}]
 \,,\qquad
 \cN_{\{J_{v}\}}
=
\bigotimes_{v\in V^{\pp}}\textrm{Inv}\Big{[}
\cV_{J_{v}}\otimes \bigotimes_{e\,| v\in\pp e} \cV_{j_{e}}
\Big{]}
\,,
\ee
since we hare working at fixed boundary spins $j_{e\in\pp}$.

For spin network basis state, a straightforward calculation leads to the multiplicity matrices $\rho_{\{J_{v}\}}$ simply given by the boundary components of the intertwiners living at the boundary vertices:
\begin{lemma}
For a spin network basis states $\Psi_{\{j_{e},I_{v}\}}$ with given spins $j_{e}$ on all bulk and boundary edges, as well as chosen intertwiner states $I_{v}$ at each vertex, we decompose the intertwiner states living on boundary vertices in the bouquet spin basis separating their ``boundary'' component from their ``bulk'' component,
\be
\forall v\in V^{\pp}\,,\quad
I_{v}=\sum_{J_{v}}
C_{v}(J_{v})\,
|J_{v},\bI_{v}^{(J_{v})},\BI_{v}^{(J_{v})}\ra\,,
\ee
with normalized  intertwiners $\bI_{v}^{(J_{v})}$ and $\BI_{v}^{(J_{v})}$, respectively between the boundary spins and the bouquet spin, then between the bouquet spin and the bulk spins attached to the vertex $v$.
Then the induced boundary density matrix reads:
\be
\rho_{\pp}[\Psi_{\{j_{e},I_{v}\}}]
=
\sum_{\{J_{v}\}_{v\in V^{\pp}}}
\bigotimes_{v\in V^{\pp}}\f{\id_{J_{v}}}{(2J_{v}+1)} \otimes
\rho_{\{J_{v}\}}\,,
\qquad\textrm{where} \quad
|\bI_{v}^{(J_{v})}\ra\in
\mathrm{Inv}\Big{[}
\cV_{J_{v}}\otimes \bigotimes_{e\,| v\in\pp e} \cV_{j_{e}}
\Big{]}
\,.
\ee
The multiplicity matrices $\rho_{\{J_{v}\}}$ have rank-one:
\be
\rho_{\{J_{v}\}}
=
|\iota_{\{J_{v}\}}\ra\la \iota_{\{J_{v}\}}|\,,\quad
\iota_{\{J_{v}\}}
=
\bigotimes_{v\in V^{\pp}}
C_{v}(J_{v})
|\bI_{v}^{(J_{v})}\ra
\,\in\cN_{\{J_{v}\}}
\,.
\ee
\end{lemma}
This rank-one property obviously extends to possible spin network superposition states with correlation between bouquet spins, i.e. with coefficients $C(\{J_{v}\})$ generalizing the factorized ansatz $\prod_{v}C_{v}(J_{v})$ of basis states, but is ruined as soon as there is  non-trivial superpositions of the bulk  components of the boundary intertwiners or more generally non-trivial intertwiner correlations between the bulk vertices.
Indeed, let us consider a generic spin network state:
\be
\psi=\sum_{\{j_{e}\},\{I_{v}\}}
C^{\{j_{e}\}_{e\in\pp\Gamma},\{j_{e}\}_{e\in\Gamma^{o}}}_{\{J_{v}\}_{v\in V^{\pp}}}(\{\bI_{v}^{(J_{v})},\BI_{v}^{(J_{v})}\}_{v\in V^{\pp}},\{I_{w}\}_{w\notin V^{\pp}})
\bigotimes_{v\in V^{\pp}}\Big{(}\bI_{v}^{(J_{v})}\otimes\BI_{v}^{(J_{v})}\Big{)}
\,\otimes\,
\bigotimes_{w\notin V^{\pp}} I_{w}
\quad\in\cH_{\Gamma}\,,
\ee
where we use the notation $v$ for the boundary vertices and $w$ for the remaining vertices of the bulk graph. We have chosen an arbitrary orthonormal basis of intertwiners $I_{w}$ for bulk vertices while using explicitly the bouquet spin basis for the boundary vertices. A straightforward calculation yields the following induced boundary density matrix:
\be
\rho_{\pp}[\psi]=
\sum_{\{J_{v}\}_{v\in V^{\pp}}}
\bigotimes_{v\in V^{\pp}}\f{\id_{J_{v}}}{(2J_{v}+1)} \otimes
\rho_{\{J_{v}\}}\,,
\ee
\be
\rho_{\{J_{v}\}}
\,=\,
\sum_{\bI_{v},\widetilde{\bI_{v}}}\sum_{j_{e},\BI_{v},I_{w}}
C^{\{j_{e}\}}_{\{J_{v}\}}(\{\widetilde{\bI_{v}^{(J_{v})}},\BI_{v}^{(J_{v})}\},\{I_{w}\})
\,\overline{C^{\{j_{e}\}}_{\{J_{v}\}}(\{\bI_{v}^{(J_{v})},\BI_{v}^{(J_{v})}\},\{I_{w}\})}\,
\bigotimes_{v\in V^{\pp}}|\widetilde{\bI_{v}^{(J_{v})}}\ra\la \bI_{v}^{(J_{v})}|
\,.
\ee
The integration over the bulk holonomies amounts in the end to the partial trace over the bulk intertwiners (i.e. the intertwiner states at the vertices not connected to any boundary edge), over the bulk component of the intertwiners at the boundary vertices, and over the spins of the graph edges. This partial trace naturally leads to mixed states on the multiplicity spaces $\cN_{\{J_{v}\}}$ with higher rank multiplicity matrices $\rho_{\{J_{v}\}}$.
This means that non-trivial bulk correlations (between bulk intertwiners and bulk spins) get reflected in the rank of the multiplicity matrices $\rho_{\{J_{v}\}}$. This is a much finer witness of the bulk structure that the overall closure defect.

This hints towards a natural layer-by-layer reconstruction of the bulk from the boundary density matrix. Starting from $\rho_{\pp}$, one can try the various partitions of the boundary, grouping the boundary edges, and check which partition leads to a multiplicity matrix with the lowest rank, and thus with the least correlation between boundary vertices. Once this first layer of the bulk graph, one would thank follow the same logic to reconstruct the second layer of the bulk, grouping the bouquets together so that the second layer intertwiners are the least correlated possible. We would pursue this onion-like reconstruction until we reach the inner loop of the universal reconstruction procedure described in the previous section. It would be enlightening if one could translate this idea of a bulk with the least correlation between graph vertices into an action principle whose extrema would determine the bulk structure from the quantum boundary data fixed by the chosen boundary density matrix.

\section{Examples: Boundary Density Matrix for Candy Graphs}

We would like to conclude this paper  with explicit examples of the bulk-to-boundary procedure, from bulk spin networks to boundary density matrices. We will consider the case of a bulk graph with two boundary vertices. The deeper bulk structure does not matter and it is enough to consider a single loop to which the bulk edges connect. We consider the two examples of boundary vertices each with two boundary edges, and then each with three boundary edges, as drawn on fig.\ref{fig:candygraphs}.
\begin{figure}[hbt!]
\centering
\vspace*{3mm}
\begin{tikzpicture}[scale=0.7]

\draw [domain=0:360,dotted,thick] plot ({2.4 * cos(\x)}, {1.2 * sin(\x)});

\coordinate (A1) at (-2,0);
\coordinate (A2) at (2,0);

\draw (0,0) node {bulk};

\draw[thick] (A1) node[scale=0.7] {$\bullet$} -- ++ (145:1.5);
\draw[thick] (A1) -- ++ (215:1.5);

\draw[thick] (A2) node[scale=0.7] {$\bullet$} -- ++ (35:1.5);
\draw[thick] (A2) -- ++ (-35:1.5);

\draw [domain=0:360,dotted,thick] plot ({12+2.4 * cos(\x)}, {1.2 * sin(\x)});

\coordinate (B1) at (10,0);
\coordinate (B2) at (14,0);

\draw (12,0) node {bulk};

\draw[thick] (B1) node[scale=0.7] {$\bullet$} -- ++ (145:1.5);
\draw[thick] (B1) -- ++ (180:1.5);
\draw[thick] (B1) -- ++ (215:1.5);

\draw[thick] (B2) node[scale=0.7] {$\bullet$} -- ++ (35:1.5);
\draw[thick] (B2) -- ++ (0:1.5);
\draw[thick] (B2) -- ++ (-35:1.5);

\end{tikzpicture}
\caption{
Candy graphs.}
\label{fig:candygraphs}
\end{figure}
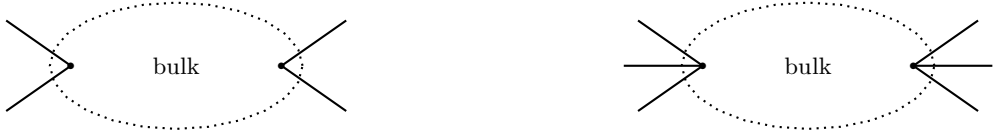

\subsection{The four-qubit candy graph}

Let us describe the graph with two vertices linked by a single loop and each with two boundary edges, as drawn on fig.\ref{fig:4candygraph}. We assume that the spin on the fours boundary edges are fixed to $j_{1}=j_{2}=j_{3}=j_{4}=\f12$ and we also fix the spin around the loop to an arbitrary value $k$. 
The bulk Hilbert space thus consists in the tensor product of the spaces of intertwiners living at the two vertices $\alpha$ and $\beta$:
\be
\cH_{bulk}= \cH_{\alpha}\otimes \cH_{\beta}\,,
\qquad
 \cH_{\alpha}= \cH_{\beta}
=
\textrm{Inv}[\cV_{{\f12}}\otimes \cV_{{\f12}}\otimes \cV_{k}\otimes \cV_{k}]
\,.
\ee
\begin{figure}[hbt!]
\centering
\begin{tikzpicture}[scale=0.7]

\coordinate (A1) at (0,0);
\coordinate (A2) at (3,0);

\draw[thick,in=115,out=65,rotate=0] (A1)  to node[above] {$k$} (A2)  node[scale=0.7] {$\bullet$} (A2) to[out=245,in=-65] node[below] {$k$} (A1) node[scale=0.7] {$\bullet$};

\draw[thick] (A1) -- ++ (145:1.5) ++ (145:0.35) node {$\f12$} ++ (15:0.5) node {$\circled{1}$};
\draw[thick] (A1) -- ++ (215:1.5) ++ (215:0.35) node {$\f12$} ++ (-15:0.5) node {$\circled{2}$};

\draw[thick] (A2) -- ++ (35:1.5) ++ (35:0.35) node {$\f12$} ++ (165:0.5) node {$\circled{3}$};
\draw[thick] (A2) -- ++ (-35:1.5) ++ (-35:0.35) node {$\f12$} ++ (-165:0.5) node {$\circled{4}$};

\draw[->,>=stealth,very thick] (6,0) -- (7.5,0);

\coordinate (B1) at (11.5,0);
\coordinate (B2) at (13.5,0);

\draw[thick,in=105,out=75,rotate=0] (B1)  to (B2)  node[scale=0.7] {$\bullet$} (B2) to[out=255,in=-75] (B1) node[scale=0.7] {$\bullet$};

\draw[thick] (B1) ++ (-1.5,0) node[scale=0.7] {$\bullet$} -- ++ (145:1.5) ++ (145:0.35);
\draw[thick] (B1) ++ (-1.5,0) -- ++ (215:1.5) ++ (215:0.35);

\draw[thick,red] (B1)  node[black,scale=0.7] {$\bullet$} -- node[red,above] {$J_{\alpha}$} ++ (-1.5,0)  node[black,scale=0.7] {$\bullet$};
\draw[thick,red] (B2)  node[black,scale=0.7] {$\bullet$} -- node[red,above] {$J_{\beta}$} ++ (1.5,0) node[black,scale=0.7] {$\bullet$};

\draw[thick] (B2) ++ (1.5,0) node[scale=0.7] {$\bullet$} -- ++ (35:1.5) ++ (35:0.35);
\draw[thick] (B2) ++ (1.5,0) node[scale=0.7] {$\bullet$} -- ++ (-35:1.5) ++ (-35:0.35);

\end{tikzpicture}
\caption{
 4-qubit candy graph: spin and intertwiner decomposition.}
\label{fig:4candygraph}
\end{figure}
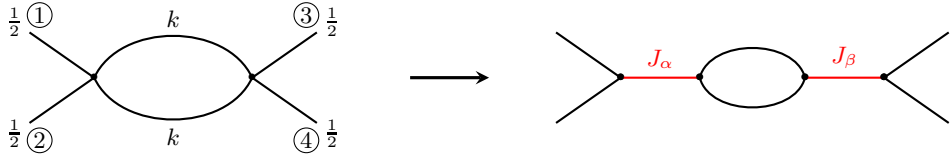

For each vertex,  $v=\alpha$ or $v=\beta$, we recouple the two boundary spins, leading to the bouquet spin basis. Here the bouquet spin $J_{v}$ can take two values, 0 or 1, and entirely determines the intertwiner state:
\be
\cH_{v}=\textrm{Inv}[\cV_{{\f12}}\otimes \cV_{{\f12}}\otimes \cV_{k}\otimes \cV_{k}]=
\C|J=0\ra\oplus \C|J=1\ra\,,\quad
\dim\cH_{v}=2
\,,
\ee
so that bulk spin network basis states are labelled by the two bouquet spins\footnotemark:
\be
\cH_{bulk}=\bigoplus_{J_{\alpha},J_{\beta}\in\{0,1\}}\C|J_{\alpha},J_{\beta}\ra
\,,\quad
\dim\cH_{bulk}=4
\,.
\ee
\footnotetext{
It might seem awkward that the dimension of the bulk Hilbert space is here (much) smaller than the dimension of the boundary Hilbert space: it would be weird to talk about a bulk-to-boundary coarse-graining in that situation. This is due to the extremely simple structure of the bulk graph. In fact, the dimension of the bulk Hilbert space increases exponentially with the number of bulk vertices (actually, more precisely, the number of independent cycles in the bulk graph as shown in \cite{Livine:2007sy}). For instance, merely pinching the loop to create an extra bulk vertices would increase the dimension of the bulk Hilbert space to $\dim\cH_{bulk}=2\times (2k+1)\times2$, which would be larger from $\dim\cH_{\pp}=2^{4}$ as soon as the spin $j$ around the loop is larger than 2.
}
The boundary Hilbert space consists in the tensor product of the four spin-$\f12$ spaces, i.e. it is made of four qubits,
\be
\cH_{\pp}=\big{(}\cV_{\f12}\big{)}^{\otimes 4}
\,,\qquad
\dim\cH_{\pp}=2^{4}\,.
\ee

Let us consider an arbitrary spin network state,
\be
\psi=\sum_{J_{\alpha},J_{\beta}}C_{J_{\alpha},J_{\beta}}|J_{\alpha},J_{\beta}\ra
\in\cH_{bulk}
\,.
\ee
The corresponding wave-function defines a boundary map, mapping the bulk holonomy along the two links of the inner loop to a boundary state in $\cH_{\pp}$:
\be
|\psi(g_{1},g_{2})\ra=\sum_{a_{i},b_{i}}
(-1)^{k-a_{1}}(-1)^{k-a_{2}}D^{k}_{a_{1}b_{1}}(g_{1})D^{k}_{a_{2}b_{2}}(g_{2})
\la (k,-a_{1})(k,-a_{2})|J_{\alpha}\ra
\la (k,b_{1})(k,b_{2})|J_{\beta}\ra
\,\in\cH_{\pp}\,.
\ee
The boundary density matrix is obtained by integrating over the bulk holonomy:
\be
\rho_{\pp}[\psi]=\int \rd g_{1}\rd g_{2}\,|\psi(g_{1},g_{2})\ra\la \psi(g_{1},g_{2})|
\,\in\textrm{End}[\cH_{\pp}]
\,.
\ee
The integration over the $\SU(2)$ group elements is straightforward to compute and yields:
\be
\rho_{\pp}[\psi]=
\sum_{J_{\alpha},J_{\beta}}|C_{J_{\alpha},J_{\beta}}|^{2}
\,\f{\id_{J_{\alpha}}}{2J_{\alpha}+1}\otimes \f{\id_{J_{\beta}}}{2J_{\beta}+1}
\,,
\ee
where $\id_{J}$, for $J=0$ and $J=1$, is the projector on the subspace of total spin $J$ in the tensor product of two qubits $(V_{{\f12}})^{\otimes 2}$.
This confirms that a pure bulk spin network state leads naturally to a mixed boundary state. Moreover, due to the simple structure of the boundary in the present example, the induced boundary density matrix  carries no entanglement between the pair of boundary edges attached to the vertex $\alpha$ and the pair attached to the vertex $\beta$.

\subsection{The six-qubit candy graph}

We can upgrade the previous  example by enriching the structure of the boundary intertwiner thereby allowing for the possibility of non-trivial entanglement between the boundary edges attached to the two vertices.
Instead of attaching two boundary edges to each vertex, we now connect three boundary edges to each vertex. We still fix the spins on the boundary edges to $j_{1}=..=j_{6}=\f12$, as well as on the inner loop to $k$ and $k+\f12$ (with the half-integer shift to account for the extra half-spin on the boundary) for $k>0$.
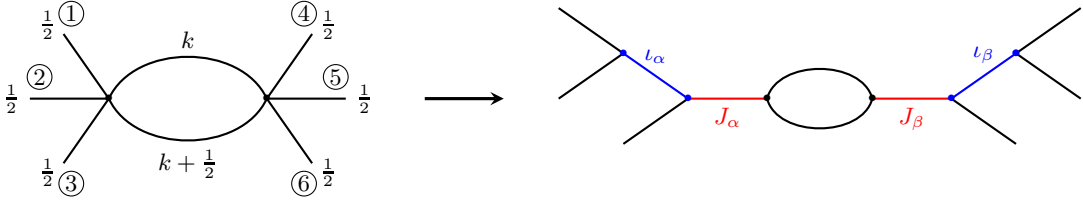
\begin{figure}[hbt!]
\centering
\begin{tikzpicture}[scale=0.7]

\coordinate (A1) at (0,0);
\coordinate (A2) at (3,0);

\draw[thick,in=115,out=65,rotate=0] (A1)  to node[above] {$k$} (A2)  node[scale=0.7] {$\bullet$} (A2) to[out=245,in=-65] node[below] {$k+\f12$} (A1) node[scale=0.7] {$\bullet$};

\draw[thick] (A1) -- ++ (125:1.5) ++ (150:0.35) node {$\f12$} ++ (25:0.5) node {$\circled{1}$};
\draw[thick] (A1) --node[very near end,above] {$\circled{2}$} ++ (180:1.5) ++ (180:0.35) node {$\f12$};
\draw[thick] (A1) -- ++ (235:1.5) ++ (210:0.35) node {$\f12$} ++ (-25:0.5) node {$\circled{3}$};

\draw[thick] (A2) -- ++ (55:1.5) ++ (30:0.35) node {$\f12$} ++ (155:0.5) node {$\circled{4}$};
\draw[thick] (A2) --node[very near end,above] {$\circled{5}$} ++ (0:1.5) ++ (0:0.35) node {$\f12$};
\draw[thick] (A2) -- ++ (-55:1.5) ++ (-30:0.35) node {$\f12$} ++ (-155:0.5) node {$\circled{6}$};

\draw[->,>=stealth,very thick] (6,0) -- (7.5,0);

\coordinate (B1) at (12.5,0);
\coordinate (B2) at (14.5,0);

\draw[thick,in=105,out=75,rotate=0] (B1)  to (B2) to[out=255,in=-75] (B1);

\draw[thick] (B1) ++ (-1.5,0) coordinate(B3);
\draw[thick] (B1) ++ (-1.5,0) -- ++ (215:1.5);

\draw[thick,blue] (B3) -- node[midway,above=1.5] {$\iota_{\alpha}$} ++ (145:1.5) ;
\draw[thick] (B3) ++ (145:1.5) -- ++ (215:1.5);
\draw[thick] (B3) ++ (145:1.5)  node[blue,scale=0.7] {$\bullet$} -- ++ (145:1.5);

\draw[thick,red] (B1) node[black,scale=0.7] {$\bullet$} -- node[red,below] {$J_{\alpha}$} ++ (-1.5,0) node[scale=0.7,blue] {$\bullet$} ;
\draw[thick,red] (B2) node[black,scale=0.7] {$\bullet$} -- node[red,below] {$J_{\beta}$} ++ (1.5,0) ;

\draw[thick] (B2) ++ (1.5,0) coordinate(B4);
\draw[thick] (B2) ++ (1.5,0) node[scale=0.7,blue] {$\bullet$}  -- ++ (-35:1.5);

\draw[thick,blue] (B4) -- node[midway,above=1.5] {$\iota_{\beta}$} ++ (35:1.5) node[scale=0.7] {$\bullet$};
\draw[thick] (B4) ++ (35:1.5) -- ++ (35:1.5);
\draw[thick] (B4) ++ (35:1.5)  node[blue,scale=0.7] {$\bullet$} -- ++ (-35:1.5);

\end{tikzpicture}
\caption{
The one-loop 6-qubit candy graph and intertwiner basis.}
\label{fig:6candygraph}
\end{figure}

The bulk Hilbert space thus consists in the tensor product of the spaces of intertwiners living at the two vertices $\alpha$ and $\beta$:
\be
\cH_{bulk}= \cH_{\alpha}\otimes \cH_{\beta}\,,\qquad
\cH_{\alpha}
= \cH_{\beta}
=
\textrm{Inv}\big{[}(\cV_{{\f12}})^{\otimes 3}\otimes \cV_{k}\otimes \cV_{k+\f12}\big{]}
\,.
\ee
For each vertex, $v=\alpha$ and $v=\beta$, we unfold the intertwiner space by recoupling the three spins $\f12$ together into the bouquet spin  $J_{v}$, as drawn on fig.\ref{fig:6candygraph}.
Since the 3-valent intertwiner between the spins $k$, $k+\f12$ and $\f12$ is unique (and given by the corresponding Clebsh-Gordan coefficients), we can put aside this bulk component of the intertwiner and focus on the boundary component of the intertwiner.
Then, since the tensor product of three  spins $\f12$ decomposes as
\be
(\cV_{{\f12}})^{\otimes 3}=\cV_{\f32}\otimes 2 \cV_{\f12}\,,
\ee
the intertwiner space is three-dimensional:
\be
\cH_{v}
=
\C|J_{v}=\f32\ra \oplus \C|J_{v}=\f12,\iota_{v}=0\ra\oplus \C|J_{v}=\f12,\iota_{v}=1\ra
\,.
\ee
The extra index $\iota\in\{0,1\}$ when the three qubits recouple to the bouquet spin $J=\f12$ label the degeneracy in the decomposition of the tensor product. As depicted on fig.\ref{fig:6candygraph}, we can simply take it as the spin recoupling for the first two qubits (boundary edges 1 and 2 for the vertex $\alpha$ and boundary edges 4 and 5 for the vertex $\beta$). In that case, we can extend the convention for the intertwiner basis state $|J,\iota\ra$ even to $J=\f32$, in which case the extra label is allowed to take a single value $\iota=1$.

Bulk spin network basis states are then defined by the choice of the two intertwiner basis states at $v=\alpha$ and $v=\beta$:
\be
\cH_{bulk}
=
\bigoplus_{\{J_{v},\iota_{v}\}}\C|J_{v},\iota_{v}\ra
\,,\quad
\dim \cH_{bulk}=3\times 3=9\,.
\ee
The boundary Hilbert space simply consists in 6 qubits, from which we also use the bouquet spin basis:
\be
\cH_{\pp}
=
\big{(}\cV_{\f12}\big{)}^{\otimes 6}
=
\cH^{\pp}_{\alpha}\otimes \cH^{\pp}_{\beta}
\,,
\quad
\cH^{\pp}_{\alpha}= \cH^{\pp}_{\beta}=
\big{(}\cV_{\f12}\big{)}^{\otimes 3}
=
\bigoplus_{J=\f12,\f32} \cV_{J}\otimes\cN_{J}\,,
\quad\textrm{with}\quad
\cN_{J}=
\textrm{Inv}\Big{[}
\cV_{J}\otimes\big{(}\cV_{\f12}\big{)}^{\otimes 3}
\Big{]}\,,
\ee
\be
\textrm{where}\quad
\dim\cN_{\f12}=2
\,,\quad
\dim\cN_{\f32}=1
\,,\quad
\dim\cH^{\pp}_{\alpha}=\dim\cH^{\pp}_{\beta}=
2\times2+4\times 1 =2^{3}\,.
\ee
Let us consider a general spin network state (with fixed spins as we have assumed so far) on this candy graph with six boundary edges:
\be
\psi=\sum_{\{J_{v},\iota_{v}\}_{v=\alpha,\beta}}
C^{J_{\alpha},J_{\beta}}_{\iota_{\alpha},\iota_{\beta}}\,|(J_{\alpha},\iota_{\alpha})\,(J_{\beta},\iota_{\beta})\ra\,.
\ee
The induced boundary density matrix, obtained after integration over the bulk holonomies, is:
\be
\rho_{\pp}[\psi]
=
\sum_{J_{\alpha},J_{\beta}}
\,\f{\id_{J_{\alpha}}}{2J_{\alpha}+1}\otimes \f{\id_{J_{\beta}}}{2J_{\beta}+1}
\otimes
\rho_{J_{\alpha},J_{\beta}}
\,,
\ee
where the multiplicity matrix encodes the data about the intertwiners:
\be 
\rho_{J_{\alpha},J_{\beta}}
\,=\,
\sum_{\{\iota_{v},\tiota_{v}\}}
C^{J_{\alpha},J_{\beta}}_{\tiota_{\alpha},\tiota_{\beta}}
\overline{C^{J_{\alpha},J_{\beta}}_{\iota_{\alpha},\iota_{\beta}}}
\Big{|}(J_{\alpha},\tiota_{\alpha})(J_{\beta},\tiota_{\beta})\Big{\ra}\Big{\la}(J_{\alpha},\iota_{\alpha})(J_{\beta},\iota_{\beta})\Big{|}
\quad\in\,\textrm{End}\big{[}\cN_{J_{\alpha}}\otimes \cN_{J_{\beta}}\big{]}
\,.
\ee
This is always a rank-one matrix and does not lead to entanglement between the boundary edges (1,2,3) and (4,5,6).

\medskip

If we want to obtain non-trivial multiplicity matrices, i.e. of higher rank, one has to allow for non-trivial bulk components of the intertwiners. To this purpose, we must consider a (slightly) more complicated bulk graph with three bulk edges connecting the two vertices. We can assume that the spins on all the edges, both on the boundary and in the bulk, are fixed to, say, $j_{1}=..=j_{9}=\f12$. If we look at the vertex $v$, which can be $\alpha$ or $\beta$, the 6-valent intertwiner can be unfolded into the bouquet spin basis. As depicted on fig.\ref{fig:6candygraph3link}, an intertwiner basis state is now labeled by the bouquet spin $J_{v}$, a multiplicity index $\iota^{\pp}_{v}\in\{0,1\}$ for the boundary component of the intertwiner (which can be taken as the recoupled spin of the edges 1 and 2) and a multiplicity index $\iota^{o}_{v}\in\{0,1\}$ for the bulk component of the intertwiner (which can be taken as the recoupled spin of the edges 4 and 5).
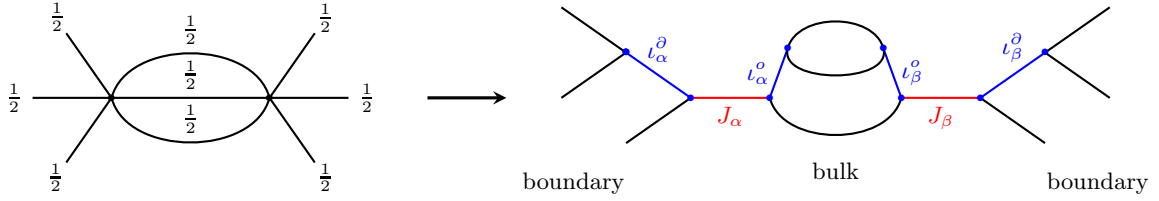
\begin{figure}[hbt!]
\centering
\begin{tikzpicture}[scale=0.7]

\coordinate (A1) at (0,0);
\coordinate (A2) at (3,0);

\draw[thick,in=105,out=75,rotate=0] (A1)  to node[above] {$\f12$} (A2)  node[scale=0.7] {$\bullet$} (A2) to[out=255,in=-75] node[above] {$\f12$} (A1) node[scale=0.7] {$\bullet$};
\draw[thick] (A1) --node[above] {$\f12$} (A2);

\draw[thick] (A1) -- ++ (125:1.5) ++ (120:0.35) node {$\f12$};
\draw[thick] (A1) -- ++ (180:1.5) ++ (180:0.35) node {$\f12$};
\draw[thick] (A1) -- ++ (235:1.5) ++ (235:0.35) node {$\f12$};

\draw[thick] (A2) -- ++ (55:1.5) ++ (55:0.35) node {$\f12$};
\draw[thick] (A2) --++ (0:1.5) ++ (0:0.35) node {$\f12$};
\draw[thick] (A2) -- ++ (-55:1.5) ++ (-55:0.35) node {$\f12$};

\draw[->,>=stealth,very thick] (6,0) -- (7.5,0);

\coordinate (B1) at (12.5,0);
\coordinate (B2) at (15,0);

\draw[thick,blue] (B1) -- node[midway,left]{$\iota_{\alpha}^{o}$} ++(70:1) coordinate (C);
\draw[thick,blue] (B2) -- node[midway,right]{$\iota_{\beta}^{o}$} ++(110:1) coordinate (D);

\draw[thick] (B2) to[out=255,in=-75] node[below=7] {bulk} (B1);

\draw[thick] (C) node[blue,scale=0.7] {$\bullet$} to[out=255,in=-75] (D) node[blue,scale=0.7] {$\bullet$} to[out=105,in=75] (C);

\draw[thick] (B1) ++ (-1.5,0) coordinate(B3);
\draw[thick] (B1) ++ (-1.5,0) -- ++ (215:1.5)++(-1,-0.75) node {boundary};

\draw[thick,blue] (B3) -- node[midway,above=1.5] {$\iota_{\alpha}^{\pp}$} ++ (145:1.5) node[scale=0.7] {$\bullet$};
\draw[thick] (B3) ++ (145:1.5)  node[blue,scale=0.7] {$\bullet$} -- ++ (215:1.5);
\draw[thick] (B3) ++ (145:1.5)  node[blue,scale=0.7] {$\bullet$} -- ++ (145:1.5);

\draw[thick,red] (B1) node[blue,scale=0.7] {$\bullet$} -- node[red,below] {$J_{\alpha}$} ++ (-1.5,0) node[scale=0.7,blue] {$\bullet$} ;
\draw[thick,red] (B2)  node[blue,scale=0.7] {$\bullet$} -- node[red,below] {$J_{\beta}$} ++ (1.5,0);

\draw[thick] (B2) ++ (1.5,0) coordinate(B4);
\draw[thick] (B2) ++ (1.5,0) node[blue,scale=0.7] {$\bullet$} -- ++ (-35:1.5) ++(1,-0.75) node {boundary};

\draw[thick,blue] (B4) -- node[midway,above=1.5] {$\iota_{\beta}^{\pp}$} ++ (35:1.5);
\draw[thick] (B4) ++ (35:1.5) -- ++ (35:1.5);
\draw[thick] (B4) ++ (35:1.5)  node[blue,scale=0.7] {$\bullet$} -- ++ (-35:1.5);

\end{tikzpicture}
\caption{
The triple link 6-qubit candy graph and intertwiner basis.}
\label{fig:6candygraph3link}
\end{figure}

The main consequence of adding bulk structure is to increase the dimension of the bulk Hilbert space:
\be
\cH_{bulk}=\cH_{\alpha}\otimes\cH_{\beta}\,,
\quad
\cH_{v}
=
\textrm{Inv}\big{[}(\cV_{\f12})^{\otimes 6}\big{]}
=
\bigoplus_{J_{v},\iota_{v}^{\pp},\iota_{v}^{o}}\C|J_{v},\iota_{v}^{\pp},\iota_{v}^{o}\ra
\,\quad
\dim\cH_{bulk}=(1+2\times 2)^{2}=25
\,.
\ee
On the other hand, the boundary Hilbert space is left unchanged. This much higher dimensionality of the bulk Hilbert space allows for finer structure of the bulk state and induced entangled on the boundary. Indeed, a generic spin network state decomposes as:
\be
\psi=
\sum_{\{J_{v},\iota_{v}\}_{v=\alpha,\beta}}
C^{J_{\alpha},J_{\beta}}_{\iota^{\pp}_{\alpha},\iota_{\alpha}^{o},\iota^{\pp}_{\beta},\iota_{\beta}^{o}}\,
\Big{|}(J_{\alpha},\iota^{\pp}_{\alpha},\iota_{\alpha}^{o})\,(J_{\beta},\iota^{\pp}_{\beta},\iota_{\beta}^{o})\Big{\ra}
\,.
\ee
Compared to the previous case of the one-loop candy graph, the bulk part of the intertwiners $\iota^{o}_{v}$ is not seen by the boundary state. This bulk data ``hidden'' from the boundary creates entanglement between the two bouquets of boundary edges. Indeed the induced boundary density matrix can be computed as:
\be
\rho_{\pp}[\psi]
=
\sum_{J_{\alpha},J_{\beta}}
\,\f{\id_{J_{\alpha}}}{2J_{\alpha}+1}\otimes \f{\id_{J_{\beta}}}{2J_{\beta}+1}
\otimes
\rho_{J_{\alpha},J_{\beta}}
\,,
\ee
where the multiplicity matrix encodes the data about the intertwiners:
\be 
\rho_{J_{\alpha},J_{\beta}}
\,=\,
\sum_{\{\iota_{v}^{\pp},\tiota_{v}^{\pp}\}}
\left(
\sum_{\{\iota_{v}^{o}\}}
C^{J_{\alpha},J_{\beta}}_{\tiota^{\pp}_{\alpha},\iota_{\alpha}^{o},\tiota^{\pp}_{\beta},\iota_{\beta}^{o}}
\overline{C^{J_{\alpha},J_{\beta}}_{\iota^{\pp}_{\alpha},\iota_{\alpha}^{o},\iota^{\pp}_{\beta},\iota_{\beta}^{o}}}
\right)\,
\Big{|}(J_{\alpha},\tiota_{\alpha})(J_{\beta},\tiota_{\beta})\Big{\ra}\Big{\la}(J_{\alpha},\iota_{\alpha})(J_{\beta},\iota_{\beta})\Big{|}
\quad\in\,\textrm{End}\big{[}\cN_{J_{\alpha}}\otimes \cN_{J_{\beta}}\big{]}
\,.
\ee
The partial trace over the bulk components of the intertwiners lead to a higher rank of the multiplicity matrix, reflecting the induced entanglement between the boundary edges attached to $\alpha$ and the ones attached to the vertex $\beta$.
A simple example is, choosing that both intertwiners have support exclusively on the bouquet spins $J_{\alpha}=J_{\beta}=\f12$, and form a Bell-like state:
\be
\psi_{Bell}=
\f1{\sqrt{2}}\,\big{(}|(\f12,0,0)(\f12,1,1)\ra-|(\f12,1,1)(\f12,0,0)\ra\big{)}\,,
\ee
leading to the induced density matrix:
\be
\rho_{\pp}[\psi_{Bell}]
=
\f{\id_{\f12}}{2}\otimes \f{\id_{\f12}}{2}\otimes \rho_{\cN}\,,\quad
\rho_{\cN}
=
|(\f12,0)(\f12,1)\ra\la(\f12,0)(\f12,1)|+|(\f12,1)(\f12,0)\ra\la(\f12,1)(\f12,0)|\,,
\ee
where the multiplicity matrix now has rank two.
This perfectly illustrates how tracing out the bulk degrees of freedom leads to a mixed state on the boundary, or in more physical terms, how correlations between bulk intertwiners leads to entanglements between boundary edges.

\section*{Conclusion \& Outlook}


In the context of the quest for understanding the holographic nature of the gravitational interaction and of quantum gravity, it is essential to investigate the bulk-boundary relation and interplay. This goes both ways: on the one hand, we need to understand the boundary modes and dynamics induced by the bulk degrees of freedom, and on the other hand, we need to understand how boundary conditions propagate within and throughout the bulk at both classical and quantum levels. Such holographic mapping between bulk and boundary theories needs to be achieved at multiple levels:  the symmetry groups, the dynamics, the quantum states, the algebra of observables.

Here, in order to start analyzing the potential holographic behavior of loop quantum gravity, we introduced explicit  2d boundaries to the 3d space, i.e. space-time corners. This 2d boundary admits a Hilbert space of boundary states, understood as quantum boundary conditions. Then loop quantum gravity's spin network states for the bulk geometry become what we call {\it boundary maps}, that is wave-functions still depending on bulk fields or degrees of freedom but valued in the boundary Hilbert space (instead of $\C$ for standard quantum mechanics). In some sense, bulk wave-functions can be interpreted as quantum circuits acting on the boundary states.
For spin network states, the bulk degrees of freedom are the $\SU(2)$ holonomies of the Ashtekar-Barbero connection along the graph links, while the boundary states are the spin states living on the spin network open edges puncturing the 2d boundary surface.
As expected, the squared norm of the bulk wave-function using the scalar product of the boundary Hilbert space gives the probability distribution for the bulk holonomies.
The new feature is that one can trace over the bulk by integrating over the bulk holonomies and obtain a density matrix for the boundary states. This {\it boundary density matrix} encodes all that we can know about the quantum state of geometry from probing the boundary if we do not have access to any bulk observable. For a pure bulk state, we typically obtain a mixed boundary state. This realizes a bulk-to-boundary coarse-graining.

Our main result is the proof that any gauge-covariant boundary density matrix for an arbitrary number of boundary edges can be induced from a pure spin network state on a simple bulk graph consisting from a single vertex connecting all the boundary edges to a single bulk loop.  In quantum information jargon, this universal reconstruction process actually purifies arbitrary mixed boundary states into pure bulk states.

\medskip

We further analyzed the algebraic structure of induced boundary density matrices, more precisely how intertwiner correlations, i.e. entanglement between bulk volume excitations, get reflected by the boundary density matrix.
This should be considered as part of the larger program of bulk tomography through boundary observables in loop quantum gravity. Hopefully, the basic tools introduced here should allow a more systematic study of how far one can see into the bulk and how much one observer on the boundary can know about the bulk spin network graph.
For instance, we would like to study in more details the relation between between boundary edge entanglement  and bulk intertwiner entanglement and quantify in a precise and explicit manner their difference.

These questions are at the kinematical level. Our hope is  more ambitious and we would like to tackle the  spin network dynamics and reformulate it in light of the bulk-boundary relation. This means projecting the bulk dynamics onto the boundary and write it in terms of boundary evolution operators. Loop quantum gravity's dynamics would then read in terms of completely positive maps\footnotemark{} acting on the boundary density matrices.
\footnotetext{
Mathematically, any evolution or measurement can be written as a completely positive map (CP map) \cite{Choi:1975,Wilde:2011npi}, which admits an operator-sum representation in terms of Kraus operators $\{ E_k, \, k=1,2,\cdots \}$:
\be \label{eq:KrausOperators}
\cE(\rho)
=
\sum_{k} E_k \, \rho \, E_k^{\dagger}
\,, \qquad
\sum_{k} E_k^{\dagger} \, E_k \leq \id
\,.
\nn
\ee
The case when $\sum_{k} E_k^{\dagger} \, E_k = \id$ are completely positive trace preserving map (CPTP map), which leave invariant the trace of quantum states. We wish to describe boundary evolution and measurements in loop quantum gravity in terms of  CPTP maps.}
Through this, the goal is to investigate in depth the implementation of the holographic principle in loop quantum gravity, and parallely move forward in the study of the coarse-graining of the theory and its renormalization flow from the Planck scale to ours. For these purposes, a general formulation in terms of boundary density matrices seems better suited to the analysis of the dynamics, measurements and coarse-graining than pure spin network states.

\section*{Acknowledgement}
Q.C. is financially supported by the China Scholarship Council.

\appendix


\bibliographystyle{bib-style}
\bibliography{LQG}

\end{document}